\documentclass{amsart}
\numberwithin{equation}{section}

\usepackage[utf8]{inputenc}
\usepackage{amssymb}
\usepackage{mathtools}
\usepackage{tikz}
\usepackage{graphicx}
\usepackage[inline]{enumitem}
\usepackage{microtype}
\usepackage[margin=1in, marginparsep=0.1in, marginparwidth=0.8in]{geometry} 
\usepackage{xcolor}
\usepackage[all]{xy}
\usepackage[textsize=footnotesize, textwidth=0.8in]{todonotes} 
\usepackage[backref=page]{hyperref}
\usepackage{hypcap}
\usepackage[capitalize, noabbrev]{cleveref}

\hypersetup{
	colorlinks,
	linkcolor={red!50!black},
	citecolor={green!50!black},
	urlcolor={blue!50!black}
}

\everyentry={\displaystyle}

\newtheorem{thm}[equation]{Theorem}
\newtheorem{lem}[equation]{Lemma}
\newtheorem{cor}[equation]{Corollary}
\newtheorem{prop}[equation]{Proposition}

\crefname{thm}{Theorem}{Theorems}
\crefname{lem}{Lemma}{Lemmas}
\crefname{cor}{Corollary}{Corollaries}

\newtheorem*{thm*}{Theorem}
\newtheorem*{thmTC}{\Cref{TC_deriv}} 
\newtheorem*{thmmain}{\Cref{mainthm}}

\theoremstyle{definition}
\newtheorem{exm}[equation]{Example}
\newtheorem{defn}[equation]{Definition}
\theoremstyle{remark}
\newtheorem{rem}[equation]{Remark}

\newcommand{\vp}{\varphi}

\newcommand{\inj}{\hookrightarrow}
\newcommand{\surj}{\twoheadrightarrow}
\newcommand{\id}{\mathrm{id}}
\newcommand{\pt}{\mathrm{pt}}

\DeclarePairedDelimiter\paren{(}{)}
\DeclarePairedDelimiter\ang{\langle}{\rangle}
\DeclarePairedDelimiter\abs{\lvert}{\rvert}
\DeclarePairedDelimiter\set{\{}{\}}
\makeatletter
	\let\oldparen\paren
	\def\paren{\@ifstar{\oldparen}{\oldparen*}}
	\let\oldbkt\bkt
	\def\bkt{\@ifstar{\oldbkt}{\oldbkt*}}
\makeatother

\newcommand{\bl}{-}
\newcommand{\term}{\textit}
\newcommand{\tikzpt}[1]{\draw[fill] #1 circle (0.5mm)}

\newcommand{\C}{\mathbb C}

\newcommand{\R}{\mathbb R}
\newcommand{\Z}{\mathbb Z}
\newcommand{\RP}{\mathbb{RP}}
\newcommand{\CP}{\mathbb{CP}}

\newcommand{\cA}{\mathcal A}
\newcommand{\cG}{\mathcal G}
\newcommand{\cH}{\mathcal H}
\newcommand{\cP}{\mathcal P}

\newcommand{\cV}{\mathcal V}

\newcommand{\fo}{\mathfrak o}

\newcommand{\MCG}{\mathrm{MCG}}
\renewcommand{\O}{\mathrm O}
\newcommand{\SL}{\mathrm{SL}}

\newcommand{\cat}{\mathsf}
\newcommand{\Bord}{\cat{Bord}}
\DeclareMathOperator{\Bun}{Bun}
\newcommand{\fC}{\cat C}
\newcommand{\fD}{\cat D}
\newcommand{\Gpd}{\cat{Gpd}}

\newcommand{\Vect}{\cat{Vect}}

\newcommand{\MTO}{\mathit{MTO}}
\newcommand{\Sq}{\mathrm{Sq}}

\DeclareMathOperator{\Aut}{Aut}
\DeclareMathOperator{\Diff}{Diff}
\DeclareMathOperator{\Cyl}{Cyl}
\renewcommand{\Im}{\operatorname{Im}}
\DeclareMathOperator{\End}{End}
\DeclareMathOperator{\Hol}{Hol}
\DeclareMathOperator{\Hom}{Hom}
\DeclareMathOperator{\Spec}{Spec}
\newcommand{\tr}{\mathrm{tr}}

\newcommand{\DS}{Z_{\mathrm{GDS}}}
\newcommand{\DW}{\DDW_0}
\newcommand{\DDW}{\mathrm{DW}}
\newcommand{\DWcl}{\mathrm{DW}^{\mathrm{cl}}}
\DeclareMathOperator{\St}{St}
\newcommand{\CSt}{\operatorname{\overline{\St}}}
\newcommand{\LDS}{L_{\mathrm{GDS}}}
\newcommand{\Zcl}{Z^{\mathrm{cl}}}
\newcommand{\PM}{\mathsf{PM}}

\newcommand{\Ptriv}{P_{\mathrm{triv}}}
\newcommand{\Cell}{\cat{Cell}}
\newcommand{\Lat}{\cat{Lat}}
\newcommand{\spin}{\mathrm{spin}}
\newcommand{\Ctop}{C_{\mathrm{top}}}
\newcommand{\Csm}{C_{\mathrm{sm}}}
\newcommand{\Line}{\cat{Line}}
\newcommand{\tH}{\widetilde H}

\begin{document}
\title[The low-energy TQFT of the GDS model]{The low-energy TQFT of the generalized double semion model}
\author{Arun Debray}
\address{The University of Texas at Austin}
\email{a.debray@math.utexas.edu}
\date{\today}
\begin{abstract}
The generalized double semion (GDS) model, introduced by Freedman and Hastings, is a lattice system similar to the
toric code, with a gapped Hamiltonian whose definition depends on a triangulation of the ambient manifold $M$, but
whose space of ground states does not depend on the triangulation, but only on the underlying manifold. In this
paper, we use topological quantum field theory (TQFT) to investigate the low-energy limit of the GDS model. We
define and study a functorial TQFT $\DS$ in every dimension $n$ such that for every closed $(n-1)$-manifold $M$,
$\DS(M)$ is isomorphic to the space of ground states of the GDS model on $M$; the isomorphism can be chosen to
intertwine the actions of the mapping class group of $M$ that arise on both sides. Throughout this paper, we
compare our constructions and results with their known analogues for the toric code.
\end{abstract}

\maketitle

\tableofcontents

\section{Introduction}


The classification of topological phases of matter is an active area of research in the theory of condensed-matter
physics and in nearby mathematical fields. There are many different approaches to this classification problem (for
an incomplete sample, see~\cite{PTBO10, LG12, CGLW13, KitaevSpectrum}), but from a mathematical point of view, a
classification via low-energy limits is appealing: based on physical insights, it is believed that the low-energy
effective theory of a gapped phase of matter is a topological quantum field theory (TQFT), possibly tensored with
an invertible theory, and that passage to the low-energy effective theory should send physically distinct phases to
distinct TQFTs~\cite{FreedHopkins, Gaiotto, RowellWang, FT18}. As TQFTs have a purely mathematical description due
to Atiyah-Segal~\cite{AtiyahTQFT, SegalCFT}, this reframes the classification question within mathematics ---
though a systematic mathematical understanding of this physical ansatz relating lattice systems to effective field
theories remains out of reach. Even at a physical level of rigor, it is not clear what the general definition of
the low-energy effective theory of a lattice model should be, and without this it is impossible to rigorously
verify the efficacy of the low-energy approach to classification in general. Nonetheless, there are many examples
of lattice models in the physical and mathematical literature, and it is instructive to study what can be said
about their low-energy effective theories in order to gain insight into the general picture. Some examples
include~\cite{Kir11, BK12, Cha14, ALW17, BCKMM, CdAIJLETS17}.

In this paper, we investigate the low-energy effective theory of the generalized double semion (GDS) lattice model of
Freedman-Hastings~\cite{FreedmanHastings}, which exists in every dimension. Freedman and Hastings define the GDS
model and study its spaces of ground states on different manifolds, showing that in even (spacetime) dimensions
$n$ they are isomorphic to the state spaces of the $\Z/2$-Dijkgraaf-Witten theory with Lagrangian equal to $0$, but
that for odd $n > 3$, they are not isomorphic to the state spaces of any $\Z/2$-Dijkgraaf-Witten theory. For every
dimension $n$, we define an $n$-dimensional TQFT $\DS\colon\Bord_n\to\Vect_\C$ and show that for every closed
$(n-1)$-manifold $M$, the state space $\DS(M)$ is isomorphic to the space of ground states of the GDS model on $M$,
and that this isomorphism is equivariant with respect to the actions of $\MCG(M)$ coming from the GDS model and the
TQFT. Along the way, we reformulate the GDS model as a lattice gauge theory with gauge group $\Z/2$: it is a theory
formulated on manifolds with a triangulation, which plays the role that a Riemannian metric does in Wick-rotated
quantum field theory. We find that, as for the toric code lattice model, the low-energy limit does not depend on
the triangulation, and is described by the state spaces of a TQFT. For both the toric code and GDS models, this
TQFT is a $\Z/2$-gauge theory, but unlike for the toric code, the GDS theory involves gravity, in that
Stiefel-Whitney classes of the underlying manifold enter the effective action. This explains the above result of
Freedman-Hastings that this TQFT cannot be any $\Z/2$-Dijkgraaf-Witten theory when $n$ is odd and greater than
$3$~\cite[Theorem 8.1]{FreedmanHastings}.

The GDS model is closely analogous to the toric code; thus, throughout this paper, we will introduce ideas first
for the toric code, which is simpler, and then turn to the GDS model. In \S\ref{lattice_models}, we define the
toric code (\S\ref{tcsect}) and GDS models (\S\ref{origGDS}) in arbitrary dimension. These are both examples of
lattice models, which are discretized analogues of quantum field theories studied in condensed-matter physics: one
puts a combinatorial structure, such as a CW structure or a triangulation, on a manifold, and formulates all data
of the theory, including the fields and the Hamiltonian, in terms of this combinatorial structure. The toric code
and GDS models are typically written as spin liquids, meaning the fields are functions from the edges of a lattice
to $\set{\uparrow, \downarrow}$. We reformulate them as lattice gauge theories, describing equivalent models whose
fields are discretizations of principal $\Z/2$-bundles.

In \S\ref{coupled_to_gravity}, we construct a class of TQFTs called $\Z/2$-gauge-gravity theories. They generalize
Dijkgraaf-Witten theories with gauge group $\Z/2$, but the Lagrangian includes Stiefel-Whitney classes of
the underlying manifold in addition to characteristic classes of the principal $\Z/2$-bundle. First, in
\S\ref{classical_action}, we define ``classical gauge-gravity theories,'' invertible TQFTs of manifolds with a
principal $\Z/2$-bundle. Then, in \S\ref{constquant}, we quantize these theories, summing over the groupoid of
principal $\Z/2$-bundles to produce TQFTs $Z_\beta\colon\Bord_n\to\Vect_\C$ of unoriented manifolds given a
cohomology class $\beta\in H^n(B\O_n\times B\Z/2;\Z/2)$.

In \S\ref{low-energy_limits}, we use these gauge-gravity TQFTs to study the low-energy behavior of the GDS model.
The Hamiltonian in the GDS model has spectrum contained within $\Z_{\ge 0}$, and the space of ground states of the
GDS model on an $(n-1)$-manifold $M$ is defined to be the kernel of the Hamiltonian for $M$. In examples arising in
physics from topological phases of matter, the space of ground states often depends only on $M$, and not on the
triangulation. When this occurs, it is expected that this extends to a TQFT $Z\colon\Bord_n\to\Vect_\C$, in that
for any closed $(n-1)$-manifold $M$, $Z(M)$ is isomorphic to the space of ground states on $M$. In \S\ref{LETC}, we
implement this idea for the toric code, where we reprove the following known result.
\begin{thmTC}
If $\DW\colon\Bord_n\to\Vect_\C$ denotes the $\Z/2$-Dijkgraaf-Witten theory with Lagrangian equal to $0$, then for
every closed $(n-1)$-manifold $M$, the space of ground states of the toric code on $M$ is isomorphic to $\DW(M)$.
\end{thmTC}
In \S\ref{GDS_deriv}, we turn to the GDS model, where we prove the main theorem. Let $\alpha\in H^1(B\Z/2;\Z/2)$
denote the generator and $w\in H^*(B\O_n;\Z/2)$ denote the total Stiefel-Whitney class.
\begin{thmmain}
Let $\beta$ be the degree-$n$ piece of $w\alpha/(1+\alpha)$. Then, for every closed $(n-1)$-manifold $M$, the space
of ground states of the GDS model on $M$ is isomorphic to $Z_\beta(M)$.
\end{thmmain}
Because of this, $Z_\beta$ will also be denoted $\DS$. Then, in \S\ref{mcg}, we strengthen
\cref{TC_deriv,mainthm} slightly: with $M$ as above, we construct actions of the mapping class group of $M$ on the
spaces of low-energy states of the toric code and GDS models on $M$, and show the isomorphisms of these spaces with
$\DW(M)$, resp.\ $\DS(M)$, are equivariant with respect to these actions.

In \S\ref{calcs}, we provide some calculations with this low-energy TQFT, allowing us to prove a comparison theorem
with $\Z/2$-Dijkgraaf-Witten theories.
\begin{thm*}\hfill
\begin{enumerate}
	\item In dimension $3$, there is an isomorphism between $\DS$ and the $\Z/2$-Dijkgraaf-Witten theory with
	Lagrangian equal to the nonzero element of $H^3(B\Z/2;\Z/2)$.
	\item\label{evenpart} In any even dimension, there is an isomorphism between $\DS$ and $\DW$.
	\item\label{oddcase} For odd $n\ge 5$, $\DS$ is distinct from all $\Z/2$-Dijkgraaf-Witten theories.
\end{enumerate}
\end{thm*}
This theorem is a combination of \cref{dim3isom,isDW,notDW}. Part~\eqref{oddcase} was first proven
by~\cite{FreedmanHastings}, as was~\eqref{evenpart} for state spaces.

\subsection*{Acknowledgements}
I gratefully thank my advisor, Dan Freed, for his constant help and guidance.

\section{The toric code and GDS models}
	\label{lattice_models}
		\begin{defn}
Let $X$ be a topological space with a CW structure $\Xi$. We let $\Delta^k(X)$ denote its set of $k$-cells and
$X^k$ denote its $k$-skeleton. When we need to make explicit that these are with respect to $\Xi$, we will write
$\Delta^k(X;\Xi)$, resp.\ $X_\Xi^k$. If $\Pi$ is a triangulation of $X$, we will also write $\Delta^k(X;\Pi)$ and
$X_\Pi^k$ for the $k$-simplices, resp.\ $k$-skeleton, of $X$ with respect to $\Pi$.

When we need $\Xi$ to be explicit, we will write $C_k^\Xi(X;A)$ (resp.\ $C^k_\Xi(X;A)$) for the group of cellular
$k$-chains (resp.\ $k$-cochains) with coefficients in an abelian group $A$ for the CW structure $\Xi$. We will
employ analogous notation for cycles and cocycles, and for simplicial (co)chains and (co)cycles with respect to a
given triangulation $\Pi$.
\end{defn}
\begin{defn}
\label{spindefn}
For a topological space $X$, let $\Bun_{\Z/2}(X)$ denote the groupoid of principal $\Z/2$-bundles on $X$, and if
$Y\subset X$, let $\Bun_{\Z/2}(X, Y)$ denote the groupoid of principal $\Z/2$-bundles $P\to X$ equipped with a
trivialization $\xi$ over $Y$.

If $X$ is a CW complex, then $(P,\xi)\in\Bun_{\Z/2}(X^1, X^0)$ determines a function $\spin_{(P,\xi)}\colon
\Delta^1(X)\to\Z/2$: if $e$ is a $1$-cell of $X$, $P|_e$ descends to a principal bundle $P'\to e/\partial e$, where
we use the trivialization of $P$ on $\partial e$ to identify the fibers. Then $\spin_{(P,\xi)}(e)$ is $0$ if $P'$
is trivial, and $1$ if it is nontrivial.
\end{defn}
In other words, if $\partial e = \set{v,w}$, we can compare $\xi(v)$ and $\xi(w)$ by parallel-transporting along
$e$; then $\spin_{(P,\xi)}(e)$ is their difference. The function $\spin_{(P,\xi)}$ determines $(P,\xi)$ up to
isomorphism.

\subsection{The toric code}
\label{tcsect}
The toric code was originally studied by Kitaev~\cite{Kitaev}. He was interested in its properties as a quantum
error-correcting code when put on a torus, hence the name ``toric code;'' a more descriptive name would be
``lattice gauge theory for a finite group $G$.'' Subsequently, it has been generalized in many directions: defining
it on nonorientable surfaces~\cite{FL01}; generalizing it to manifolds of any dimension~\cite{FML}; placing the
spins on $k$-cells, rather than edges~\cite{DKLP}; considering a fermionic variant~\cite{fermionicTC}; changing
whether it is even a gauge theory at all~\cite{BMCA}; and adding global symmetries~\cite{BBJCW16, HBFL16, LV16}. In
this paper, we will not consider most of these generalizations.

Fix a dimension $n$, which will always be the \term{spacetime dimension}; that is, lattice models are on
$(n-1)$-manifolds, and TQFTs are formulated with $n$-dimensional cobordisms between $(n-1)$-dimensional manifolds.
The toric code assigns to a closed $(n-1)$-manifold $M$ together with a CW structure a finite-dimensional complex
vector space $\cH$, called the \term{state space}, and a self-adjoint operator $H\colon\cH\to\cH$, called the
\term{Hamiltonian}. We proceed to define these.

The groupoid of fields for the toric code is $\Bun_{\Z/2}(M^1, M^0)$, and the state space assigned to $M$ is
$\cH\coloneqq \C[\Bun_{\Z/2}(M^1, M^0)]$, the vector space of complex-valued functions on the groupoid of fields.
Given $(P,\xi)\in\pi_0\Bun_{\Z/2}(M^1, M^0)$, let $\delta_{(P,\xi)}\in\cH$ be the function sending
$(P,\xi)\mapsto 1$ and all nonisomorphic $(P',\xi')$ to $0$. The set
\begin{equation}
\label{orthbasis}
	\set{\delta_{(P,\xi)}\mid (P,\xi)\in\pi_0\Bun_{\Z/2}(M^1, M^0)}
\end{equation}
is a basis for $\cH$; endow $\cH$ with the inner product for which it is an orthonormal basis.

Given a $0$-cell $v$ of $M$, let $A_v\colon\cH\to\cH$ denote the \term{shift operator} at $v$: if $\psi\in\cH$ and
$(P,\xi)\in\Bun_{\Z/2}(M^1, M^0)$, let $\xi+\delta_v$ denote the section of $P$ on $M^0$ which is identical to
$\xi$ except on $v$, where its value is $\xi(v)+1$. Then,
\begin{subequations}
\label{A_vB_f}
\begin{equation}
\label{A_vB_fA}
	A_v(\psi)(P,\xi) \coloneqq \psi(P, \xi+\delta_v).
\end{equation}
Given a $2$-cell $f$ of $M$, let $B_f\colon\cH\to\cH$ be multiplication by the holonomy around $\partial f$:
\begin{equation}
\label{A_vB_fB}
	B_f(\psi)(P,\xi) \coloneqq (-1)^{\Hol_P(f)}\psi.
\end{equation}
\end{subequations}
There are operators associated to each $2$-cell $f$ and each $0$-cell $v$, called \term{face operators}, resp.\
\term{vertex operators}:
\begin{subequations}
\label{geom_TC}
\begin{align}
\label{geomTCface}
H_f &\coloneqq \frac{1 - B_f}{2}\\
H_v &\coloneqq \frac{1 - A_v}{2},
\label{geomTCvertex}
\end{align}
\end{subequations}
and the Hamiltonian assigned to $M$ is
\begin{equation}
\label{geomTCham}
H_{\mathrm{TC}} \coloneqq \sum_{v\in\Delta^{0}(M)} H_v + \sum_{f\in\Delta^{2}(M)} H_f.
\end{equation}
\begin{rem}
The original definition of the toric code looked different, replacing $(P,\xi)$ with the function
$\spin_{(P,\xi)}\colon\Delta^1(M)\to\Z/2$ it defines. The state space is the free complex vector space on the
finite set of these functions. The analogues of $A_v$ and $B_f$ for $v\in\Delta^{0}(M)$ and $f\in\Delta^{2}(M)$ are
\begin{subequations}
\label{origops}
\begin{align}
	A_v' &\coloneqq \prod_{e: v\in\partial e} \sigma_e^x\\
	B_f' &\coloneqq \prod_{e\in\partial f} \sigma_e^z.
\end{align}
\end{subequations}
Here, $\sigma^x$ and $\sigma^z$ are the Pauli operators
\begin{equation}
\label{paulidesc}
	\sigma^x = \begin{pmatrix}0 & 1\\1 & 0\end{pmatrix},\qquad\qquad\qquad\sigma^z = \begin{pmatrix*}[r]1 & 0\\0 &
	-1\end{pmatrix*}.
\end{equation}
The state space $\cH$ can be identified with the tensor product of local state spaces
$\cH_e\coloneqq\C\cdot\set{0,1}$ over each $1$-cell $e$, and the notation $\sigma_e^x$ and $\sigma_e^z$ means these
operators act on $\cH_e$ by the matrices in~\eqref{paulidesc}, and by the identity on the remaining tensor factors.

We can identify $A_v'$ with $A_v$ by observing that switching the trivialization for $(P,\xi)$ over $v$ amounts to
switching the value of $\spin_{(P,\xi)}$ on any 1-cell $e$ adjacent to $v$, which is the action by $\sigma_e^x$. To
identify $B_f$ and $B_f'$, observe that the holonomy of $(P,\xi)$ around $\partial f$ is the product of the spins
on the $1$-cells in $\partial f$.
\end{rem}
\begin{prop}\hfill
\label{TC_lin_alg}
\begin{enumerate}
	\item\label{TC_self_adj} The Hamiltonian $H_{\mathrm{TC}}$ is self-adjoint.
	\item\label{TC_commute} The $H_f$ and $H_v$ operators are projectors, and pairwise commute.
	\item\label{gappedTC} $\Spec(H_{\mathrm{TC}})\subset \Z_{\ge 0}$, and $0$ is always an eigenvalue.
\end{enumerate}
\end{prop}
\begin{proof}[Proof sketch]
Using the identifications of $A_v$ with $A_v'$ and $B_f$ with $B_f'$, $A_v$ and $B_f$ are products of real
symmetric matrices, hence are themselves real symmetric matrices; therefore $H_v$ and $H_f$ are too. Therefore $H$
is a sum of real symmetric matrices, proving part~\eqref{TC_self_adj}.

Part~\eqref{TC_commute} is directly analogous to Kitaev's original proof in dimension $n-1 = 2$~\cite{Kitaev};
see~\cite{FML} for the generalization to higher dimensions.

Part~\eqref{gappedTC} follows because the eigenvalues of $A_f$ and $B_v$ are in $\set{\pm 1}$, so the eigenvalues
of $H_f$ and $H_v$ are in $\set{0,1}$. The trivial bundle, together with the identity trivialization, is an
eigenvector for $0$.
\end{proof}
\subsection{Generalized double semion model}
\label{origGDS}
Our main focus is the generalized double semion (GDS) model.

The double semion model for $n = 3$ was first studied by Freedman-Nayak-Shtengel-Walker-Wang~\cite{FNSWW} and
Levin-Wen~\cite[\S VI.A]{LevinWen}, then generalized to all dimensions $n$ by Freedman and
Hastings~\cite{FreedmanHastings}.\footnote{There are a few other generalizations of the double semion model in low
dimensions~\cite{vKBS, LV16, OMD, DOVM18}, but we focus on Freedman-Hastings' construction.} The name comes from
the excitations in the $n = 3$ case, which produce pairs of semions, anyonic quasiparticles with statistics
intermediate between those of bosons and fermions.\footnote{The name ``generalized double semions'' is somewhat of
a misnomer, however: anyons cannot exist in dimension $n > 3$, because the braids that define their mutual
statistics can be unlinked. See~\cite[\S2.1]{RowellWang}. It is also not clear that the theory is the double of
another~\cite[\S1]{FreedmanHastings}. At least it is generalized.}
%
%
\begin{defn}
Let $M$ be a simplicial complex and $c$ be a simplex of $M$.
\begin{itemize}
	\item The \term{open star} of $c$, denoted $\St(c)$, is the subset of $M$ consisting of all simplices whose
	closures contain $c$.
	\item The \term{closed star} of $c$, denoted $\CSt(c)$, is the smallest subcomplex containing $\St(c)$).
\end{itemize}
\end{defn}
For the GDS model, we need a neighborhood of $v$ in between the open and closed stars of $v$.
\begin{defn}
Let $M$ be a simplicial complex and $e$ be a simplex of $M$. Define the \term{$0$-clopen star} $\CSt(0)(e)$ to be
$\St(e)\cup\CSt(e)^0$. That is, we include the $0$-simplices of the closed star of $e$ as well as all cells in the
open star.
\end{defn}
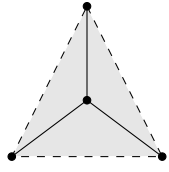
\begin{figure}[h!]
\begin{tikzpicture}
\coordinate (r0) at (0, -0.25);
\coordinate (s0) at (0, 1);
\coordinate (s1) at (1, -1);
\coordinate (s2) at (-1, -1);

\filldraw[fill=gray!20, draw=white] (s0) -- (s1) -- (s2) -- cycle;
\draw (r0) -- (s0);
\draw (r0) -- (s1);
\draw (r0) -- (s2);

\draw[dashed] (s0) -- (s1) -- (s2) -- cycle;

\tikzpt{(r0)};
\tikzpt{(s0)};
\tikzpt{(s1)};
\tikzpt{(s2)};
\end{tikzpicture}
\caption{The $0$-clopen star of a vertex in a simplicial structure on a surface.}
\end{figure}

As before, fix a dimension $n$; we proceed to define the state space and Hamiltonian that the GDS model assigns.
In order to avoid pathologies, one cannot define the GDS model for an arbitrary CW structure.
\begin{defn}
A \term{triangulation} of a smooth manifold $M$ is a simplicial complex $K$ together with a homeomorphism $f\colon
\abs K\to M$; if for every simplex $e$ of $K$, the restriction of $f$ to $\abs e$ is smooth, we say $(K, f)$ is a
\term{smooth triangulation}.
\end{defn}
When defining the GDS model, we choose a smooth triangulation $\Pi$ such that the $0$-clopen star of every vertex
is contractible.\footnote{The second constraint can always be satisfied after a refinement.} We discuss in
\cref{noCW} why restricting to triangulations is necessary.

The GDS model assigns to every closed $(n-1)$-manifold $M$ with such a triangulation a state space and Hamiltonian,
like the toric code does; the state space is $\C[\Bun_{\Z/2}(M^1, M^0)]$ as for the toric code, and we proceed to
define the Hamiltonian, which is similar to that of the toric code, but with an extra sign.
\begin{defn}
Let $M$ be a closed $(n-1)$-manifold with a smooth triangulation such that the $0$-clopen star of every vertex is
contractible. Then, given $(P,\xi)\in\Bun_{\Z/2}(M^1, M^0)$ and a $0$-simplex $v$, there is a unique maximal
extension of $\xi$ to a subset of $\CSt(0)(v)$; we denote that subset $Y_v'(P,\xi)$.
\end{defn}
\begin{defn}
Let $v\in\Delta^0(M;\Pi)$ and $S(v)$ denote the link of $v$ in the barycentric subdivision $\Pi_1$ of $\Pi$. Though
$S(v)$ comes equipped with a triangulation $\Pi_1|_{S(v)}$, we define a new triangulation $\Pi_{S(v)}$ on $S(v)$.
For $k\ge 0$, if $e$ is a $(k+1)$-simplex of $\Pi$ such that $v\in\partial e$, let
\begin{equation}
	C(e)\coloneqq \set{c\in\Delta^*(S(v), \Pi_1|_{S(v)}): \abs c\subset\abs e}.
\end{equation}
For each such $e$, we define a $k$-simplex of $\Pi_{S(v)}$, denoted $S(v)\cap e$, whose geometric realization is
\begin{equation}
	\abs{S(v)\cap e} \coloneqq \bigcup_{c\in C(e)} c.
\end{equation}
We say that $S(v)\cap e'$ is a face of $S(v)\cap e$ if every $c'\in C(e')$ is a face of some $c\in C(e)$, which may
depend on $c'$. This data defines a triangulation on $S(v)$ such that if $e$ is a simplex of $\Pi$ with
$v\in\partial e$,
\begin{equation}
	\abs{S(v)\cap e} = \abs{S(v)}\cap\abs e.
\end{equation}
\end{defn}
From now on, the triangulation on $S(v)$ is assumed to be $\Pi_{S(v)}$ unless stated otherwise.
\begin{defn}
\label{myGDSsignchange}
Let $(P,\xi)\in\Bun_{\Z/2}(M^1, M^0)$. For any $v\in\Delta^{0}(M)$, let
\begin{equation}
	Y_v(P,\xi)\coloneqq \set{S(v)\cap e\mid e\in Y_v'(P,\xi)},
\end{equation}
which is a subcomplex of $S(v)$. The \term{GDS sign}~\cite[\S4]{FreedmanHastings} is
\begin{equation}
\label{GDS_sign}
	\sigma(v, (P,\xi))\coloneqq (-1)^{1 + \chi(\abs{Y_v(P,\xi)})}.
\end{equation}
Here $\chi$ denotes the Euler characteristic.
\end{defn}
Let $U_v$ denote the operator on $\cH$ defined by $U_v(\psi)(P,\xi) \coloneqq \sigma(v, (P,\xi))A_v(\psi)$, where
$A_v$ is as in~\eqref{A_vB_fA}. The Hamiltonian for the GDS model is
\begin{equation}
\label{eqham}
	H_{\mathrm{GDS}}\coloneqq \sum_{v\in\Delta^{0}(M)} \widetilde H_v + \sum_{f\in\Delta^{2}(M)} H_f,
\end{equation}
where $H_f$ is as in~\eqref{geomTCface} and
\begin{equation}
	 \widetilde H_v = \frac{1 - U_v}{2}.
\end{equation}
As for the toric code, we call $\widetilde H_v$ a \term{vertex operator} and $H_f$ a \term{face operator}.
\begin{rem}
\label{dualremark}
In our analysis of the GDS model, we will need to make use of the \term{dual cell complex} $\Pi^\vee$ to the
specified triangulation $\Pi$, a CW complex on $M$ with several nice properties.
\begin{itemize}
	\item $\Pi^\vee$ comes with data of a bijection $(\cdot)^\vee\colon\Delta^k(M, \Pi)\to\Delta^{n-1-k}(M,
	\Pi^\vee)$, sending a simplex to its \term{dual cell}, and such that if $e\in\partial f$, then
	$f^\vee\in\partial e^\vee$, and conversely.
	\item The map $(\cdot)^\vee$ induces a chain map on the cellular chain complexes of $\Pi$ and $\Pi^\vee$ which
	induces Poincaré duality for the cohomology of $M$ with $\Z/2$ coefficients.
	\item Each cell in $\Pi^\vee$ is a union of cells of the barycentric subdivision $\Pi_1$ of $\Pi$. (One might
	think of $\Pi_1$ as a refinement of $\Pi^\vee$; though this is not strictly true, as $\Pi^\vee$ might not come
	from a triangulation, it is a useful piece of intuition.) In particular, $\Pi^\vee$ is a regular CW complex,
	meaning the closure of each cell is contractible.
\end{itemize}
This complex is unique up to equivalence of CW complexes. Proofs of these facts follow from the results
in~\cite[\S1.6]{PLbook}.

We will also denote $((\cdot)^\vee)^{-1}$ by $(\cdot)^\vee$, but since we do not confuse $\Pi$ and $\Pi^\vee$, the
meaning will be clear from context. If $S$ is a set of cells, we write $S^\vee\coloneqq\set{e^\vee\mid e\in S}$.
\end{rem}
\begin{rem}
\label{dualFH}
Freedman-Hastings~\cite{FreedmanHastings} study a dual version of the GDS model, in that our model for $M$ and
$\Pi$ corresponds to their model for $M$ and $\Pi^\vee$. Here we compare the two setups.

Let $(P,\xi)\in\Bun_{\Z/2}(M^1,M^0)$, which defines a function $\spin_{(P,\xi)}\colon\Delta^1(M,\Pi)\to\Z/2$ as in
\cref{spindefn}; we also let $\spin_{(P,\xi)}$ denote the function $\Delta^{n-2}(M, \Pi^\vee)\to\Z/2$ defined by
precomposing with $(\cdot)^\vee$.

For any $v\in\Delta^0(M,\Pi)$, let
\begin{equation}
T(v,(P,\xi))\coloneqq \spin_{(P,\xi)}^{-1}(0)\cap \partial v^\vee,
\end{equation}
which is a closed union of cells of $\Pi^\vee$.

The GDS sign as defined by Freedman-Hastings~\cite[\S4]{FreedmanHastings} is
\begin{equation}
	\sigma'(v, (P,\xi)) \coloneqq (-1)^{1+\chi(T(v,(P,\xi))}.
\end{equation}
Let $e\in\CSt(0)(v)$. Unwinding the definitions, $e\cap S(v)\in Y_v(P,\xi)$ if and only if $e^\vee$ is a cell of
$T(v,(P,\xi))$, so the number of simplices in $Y_v(P,\xi)$ equals the number of cells in $T(v,(P,\xi))$. Since both
$T(v, (P,\xi))$ and $Y_v(P,\xi)$ are closed subsets of $M$ that are unions of cells, their Euler characteristics
are equal, so $\sigma =\sigma'$. This means there is an isomorphism between the state spaces of the model we define
above and the model as defined by Freedman-Hastings, and this isomorphism intertwines their Hamiltonians, so on any
closed $(n-1)$-manifold, the spaces of ground states of these two models are isomorphic.
\end{rem}
Next, we prove analogues of \cref{TC_lin_alg} for the GDS model. In view of \cref{dualFH}, these also follow from
results of Freedman-Hastings~\cite[Lemmas 4.1, 4.2]{FreedmanHastings}, but are proven in a different way.
\begin{lem}
\label{descent_lin_alg}
The Hamiltonian $H_{\mathrm{GDS}}$ is self-adjoint, and $\Spec(H_{\mathrm{GDS}})\subset\Z_{\ge 0}$.
\end{lem}
\begin{proof}
The first part is true because the Hamiltonian is a sum of real symmetric matrices in a basis of
$\delta$-functions, just as in the proof of \cref{TC_lin_alg}. For the second part, since the eigenvalues of $A_v$
and $B_f$ lie in $\set{\pm 1}$ and $\sigma$ is valued in $\set{\pm 1}$, then the eigenvalues of $H_f$ and $\tH_v$
lie in $\set{0,1}$.
\end{proof}
Unlike for the toric code, it is not true that $0$ is always an eigenvalue. \Cref{mainthm,simpconn} together imply
this happens for $M = \CP^{2k}$.

\begin{lem}
\label{commuting_descent}
All face operators commute, and all face operators commute with all vertex operators. After restricting to the
intersection of the kernels of the face operators, $[U_{v_1}, U_{v_2}] = 0$ and hence all vertex operators commute
when restricted to that intersection.
\end{lem}
\begin{proof}
The face operators are the same as in the toric code, hence commute by \cref{TC_lin_alg}. Operators corresponding
to simplices not in each others' closed stars commute. Therefore we have two things left to prove:
\begin{enumerate}
	\item\label{vfcomm} Given a $2$-simplex $f$ and a $0$-simplex $v\in\partial f$, $[H_f, \tH_v] = 0$.
	\item\label{vvcomm} Given a $1$-simplex $e$ and two $0$-simplices $v_1,v_2\in\partial e$, $[U_{v_1}, U_{v_2}] =
	0$ when restricted to $\bigcap_{f\in\Delta^2(M)} H_f$.
\end{enumerate}
For part~\eqref{vfcomm}: since the GDS sign factors out of $[B_f, U_v]$, then $[B_f, U_v] = \pm[B_f, A_v] = 0$ by
\cref{TC_lin_alg}, and therefore $[H_f, \tH_v] = 0$.

For part~\eqref{vvcomm}, choose $\psi\in\cH$ such that $H_f\psi = 0$ for all $2$-simplices $f$, and choose
$(P,\xi)\in\Bun_{\Z/2}(M^1, M^0)$. Since $B_f$ acts by multiplication by the holonomy of $P$ around $\partial f$,
then $\psi(P,\xi) = 0$ unless $\Hol_P(f) = 0$ for all $f$; equivalently, $P$ must extend to all of $M$.\footnote{We
will return to this point in \S\ref{LETC}.} (This extension is necessarily unique up to isomorphism.) If this is
the case,
\begin{equation}
\begin{aligned}
	[U_{v_1}, U_{v_2}]\psi(P,\xi) =\phantom{-}\sigma(v_2, (P, \xi+\delta_{v_1}))\sigma(v_1,
	(P,\xi)) \psi(P,\xi+\delta_{v_1} + \delta_{v_2}) &\\
	\phantom{a}-\sigma(v_1, (P, \xi+\delta_{v_2}))\sigma(v_2, (P,\xi))\psi(P,\xi+\delta_{v_1} + \delta_{v_2})&,
	\end{aligned}
\end{equation}
so it suffices to show that if $(P,\xi)\in\Bun_{\Z/2}(M, M^0)$,
\begin{equation}
\label{sufficesAB}
\sigma(v_2, (P, \xi+\delta_{v_1}))\sigma(v_1, (P,\xi)) =\sigma(v_1, (P, \xi+\delta_{v_2}))\sigma(v_2, (P,\xi)).
\end{equation}
Tracing through the definition of the GDS sign, this is equivalent to
\begin{equation}
\label{notCWgen}
	\chi(\abs{Y_{v_2}(P, \xi+\delta_{v_1})}) + \chi(\abs{Y_{v_1}(P, \xi)}) \underset{\bmod
	2}{\equiv} \chi(\abs{Y_{v_1}(P, \xi+\delta_{v_2})}) + \chi(\abs{Y_{v_2}(P, \xi)}).
\end{equation}
Suppose $\spin_{(P,\xi)}(e) = 0$. For $i = 1,2$, let $A(v_i)$ denote the set of simplices in $Y_{v_i}(P, \xi)$
contained in the closure of a simplex whose closure also contains $e$. Let $B(v_i)\coloneqq Y_{v_i}(P,
\xi)\setminus A(v_i)$. Then
\begin{subequations}
\label{decompintoAB}
\begin{align}
	\label{decompAB1}
	\chi(\abs{Y_{v_2}(P, \xi+\delta_{v_1})}) + \chi(\abs{Y_{v_1}(P, \xi)}) &= \#\paren{A(v_1) \amalg B(v_1)
	\amalg B(v_2)}\\
	\chi(\abs{Y_{v_1}(P, \xi+\delta_{v_2})}) + \chi(\abs{Y_{v_2}(P, \xi)}) &= \#\paren{A(v_2) \amalg B(v_2)
	\amalg B(v_1)}.
	\label{decompAB2}
\end{align}
\end{subequations}
It therefore suffices to prove that $\#{A(v_1)} = \#{A(v_2)}$. Let $c_1$ be a $1$-simplex in $A(v_1)$. Since
$2$-simplices are triangles, there exists a unique $1$-simplex $c_2$ whose closure contains $v_2$ and such that
there is a $2$-simplex $f$ with $\partial f = c_1 + c_2 + e$. By assumption, $\spin_{(P,\xi)}(e) =
\spin_{(P,\xi)}(c_1) = 0$, and since the holonomy of $P$ around $\partial f$ vanishes, $\spin_{(P,\xi)}(c_2) = 0$
too. Similarly, suppose $c_1'$ and $c_2'$ are $1$-simplices such that $v_1$ is a face of $c_1'$, $v_2$ is a face of
$c_2'$, $\spin_{(P,\xi)}(c_1') = 1$, and there is a 2-simplex $f'$ with $\partial f' = c_1' + c_2' + e$; then
$\spin_{(P,\xi)}(c_2') = 1$ too. This argument is obviously symmetric in $v_1$ and $v_2$.

%
The case $\spin_{(P,\xi)}(e) = 1$ is analogous.
\end{proof}
\begin{rem}
\label{noCW}
The ideas that go into the GDS model still make sense when one generalizes to smooth manifolds with regular CW
structures, rather than smooth triangulations, but \cref{commuting_descent} does not generalize. See \cref{noCWfig}
for a counterexample.

\begin{figure}[h!]
\begin{tikzpicture}
\draw (-2, 0) -- (2, 0) -- (2, 2) -- (-2, 2) -- cycle;
\draw (0, -1) -- (0, 3);

\node[above left] at (0, 2) {$v_1$};
\node[below left] at (0, 0) {$v_2$};

\node[left] at (-2, 1) {$0$};
\node[right] at (2, 1) {$1$};
\node[below] at (-1, 0) {$0$};
\node[below] at (1, 0) {$1$};
\node[above] at (-1, 2) {$0$};
\node[above] at (1, 2) {$0$};
\node[right] at (0, 1) {$0$};
\node[left] at (0, -1) {$1$};
\node[left] at (0, 3) {$0$};

\draw (0,0) circle (6mm);
\draw (0,0) circle (7.5mm);
\draw (0,2) circle (6mm);
\draw[blue, line width=0.5mm, domain=90:180] plot ({0.6*cos(\x)}, {0.6*sin(\x)});
\fill[blue] (0, 0.6) circle (0.7mm);
\fill[blue] (-0.6, 0) circle (0.7mm);
\draw[blue, line width=0.5mm, domain=0:180] plot ({0.6*cos(\x)}, {2+0.6*sin(\x)});
\fill[blue] (0.6, 2) circle (0.7mm);
\fill[blue] (-0.6, 2) circle (0.7mm);

\fill (0, 0) circle (0.7mm);
\fill (0, 2) circle (0.7mm);
\fill (2, 0) circle (0.7mm);
\fill (2, 2) circle (0.7mm);
\fill (-2, 0) circle (0.7mm);
\fill (-2, 2) circle (0.7mm);

\draw[red, line width=0.5mm] (0,2) circle (7.5mm);
\fill[red] (-0.75, 0) circle (0.7mm);
\end{tikzpicture}
\caption{\Cref{commuting_descent} does not generalize from triangulations to CW structures. The straight lines in
this figure depict a neighborhood on a smooth surface $\Sigma$ with a CW structure. Choose
$(P,\xi)\in\Bun_{\Z/2}(\Sigma^1, \Sigma^0)$ such that the number on each pictured $1$-cell $e$ is
$\spin_{(P,\xi)}(e)$. The circles around the $0$-cells $v_1$ and $v_2$ represent two copies each of the links
$S(v_1)$ and $S(v_2)$. The red region (shaded portions of the outer circles) is $\abs{Y_{v_1}(P,\xi)} \amalg
\abs{Y_{v_2}(P,\xi+\delta_{v_1})}$, and the blue region (shaded portions of the inner circles) is
$\abs{Y_{v_2}(P,\xi)} \amalg \abs{Y_{v_1}(P,\xi+\delta_{v_2})}$. By inspection, the Euler characteristics of these
two regions are not equal mod $2$, so~\eqref{notCWgen} does not hold in this setting, and therefore
\cref{commuting_descent} also does not apply to this CW structure: $\tH_{v_1}$ and $\tH_{v_2}$ do not commute even
when restricted to $\bigcap_f H_f$.}
\label{noCWfig}
\end{figure}
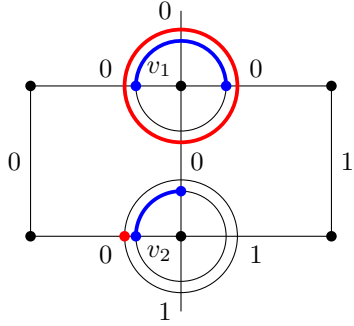

If one lets $n = 3$ and passes to the dual CW structure as in \cref{dualremark}, this recovers a fact known to
condensed-matter theorists: the double semion model on a surface can be formulated on a hexagonal lattice (or more
generally a trivalent lattice), but has an ambiguity when placed on a square lattice~\cite[\S2]{FreedmanHastings}.
This is because the dual CW structure to a trivalent lattice has triangular $2$-cells, but the dual of a
tetravalent lattice does not. For general $n$, this obstruction is encoded in the genericity assumption placed on
the CW structure in Freedman-Hastings' construction~\cite[\S4]{FreedmanHastings}; in our model this corresponds to
the restriction to smooth triangulations.
\end{rem}
\begin{lem}
\label{DS_proj}
The face operators are projectors. The operator $U_v$ has order $2$, and hence $\tH_v$ is a projector.
\end{lem}
\begin{proof}
The face operators are the same as in the toric code, hence are projectors by \cref{TC_lin_alg}. For $U_v$, choose
a $0$-simplex $v$, $\psi\in\cH$, and $(P,\xi)\in\Bun_{\Z/2}(M^1, M^0)$; then,
\begin{equation}
\label{Uv2}
	U_v^2\psi(P,\xi) = \sigma(v, (P,\xi+\delta_v)\sigma(v, (P,\xi)) \psi(P,\xi) = (-1)^{\chi(\abs{Y_v(P, \xi +
	\delta_v)}) + \chi(\abs{Y_v(P,\xi)})} \psi(P,\xi).
\end{equation}
Unwinding the definition of $Y_v$, and using that $\chi(S(v)) \equiv 0\bmod 2$, $\chi(\abs{Y_v(P,\xi+\delta_v)})
+ \chi(\abs{Y_v(P,\xi)})$ is equal mod 2 to the number of simplices $e$ in $S(v)$ of dimension at least $1$ such that
$\overline e$ contains a $1$-simplex on which $\xi$ extends and a $1$-simplex on which $\xi+\delta_v$ extends
(equivalently, on which $\xi$ does not extend). Let $Q$ be the set of such $e$.

Endow $S(v)$ with the Poincaré dual CW structure $\Pi_{S(v)}^\vee$ to the triangulation $\Pi_{S(v)}$, as in
\cref{dualremark}. Let $R\subset\Pi_{S(v)}$ be the set of $1$-simplices on which $\xi$ extends; then,
$\abs{R^\vee}$ is a topological submanifold (with boundary) of $S(v)$, and $\partial\abs{R^\vee} = \abs{ Q^\vee}$.
Hence $\chi(\abs{Q^\vee}) \equiv 0\bmod 2$; since $Q^\vee$ is a subcomplex of $\Pi_{S(v)}^\vee$, this means
$Q^\vee$ has an even number of cells, so $Q$ has an even number of simplices. Thus $\chi(\abs{Y_v(P,\xi+\delta_v)})
+ \chi(\abs{Y_v(P,\xi)})\equiv 0\bmod 2$, and this suffices by~\eqref{Uv2}.
\end{proof}

There are a few other equivalent ways to define the GDS sign. We record one which we will use later.
\begin{prop}
\label{alternatesign}
Let $(P,\xi)\in\Bun_{\Z/2}(M^1, M^0)$ and $v\in\Delta^0(M)$, and let $N_v$ be the set of simplices $c$ of $M$ with
$v\in\partial c$. If $Z_v(P,\xi)\subset N_v$ denotes the subset of simplices $c$ such that either
\begin{enumerate*}
	\item $c$ is a $1$-simplex and $\spin_{(P,\xi)}(c) = 1$, or
	\item there is a $1$-simplex $e\in\partial c$ with $\spin_{(P,\xi)}(e) = 1$,
\end{enumerate*}
then $(-1)^{1+\#{Z_v(P,\xi)}} = \sigma(v, (P,\xi))$.
\end{prop}
\begin{proof}
It suffices to show $\#{Z_v(P,\xi)} \equiv \#{Y_v(P,\xi)}\bmod 2$. If $W_v(P,\xi)$ denotes the subset of $N_v$
consisting of simplices $c$ such that either
\begin{enumerate*}
	\item $c$ is a $1$-simplex and $\spin_{(P,\xi)}(c) = 0$, or
	\item $\spin_{(P,\xi)}(e) = 0$ for all $e\in\Delta^1(\partial c)$,
\end{enumerate*}
then the map $c\mapsto c\cap S(v)$ for $c\in N_v$ restricts to a bijection from $W_v(P,\xi)$ to $Y_v(P,\xi)$.

By definition, $Z_v(P,\xi)$ is the complement of $W_v(P,\xi)$ inside $N_v$. Since $N_v^\vee = \partial v^\vee$ and
$\chi(\abs{\partial v^\vee})$ is even, then $\#{N_v}$ is even and
\begin{equation}
	\#{Z_v(P,\xi)} + \#{Y_v(P,\xi)} = \#{Z_v(P,\xi)} + \#{W_v(P,\xi)} = \#{N_v}\equiv 0\bmod 2.
	\qedhere
\end{equation}
\end{proof}

\section{Gauge-gravity TQFTs}

%

\label{coupled_to_gravity}
As part of our goal of studying the low-energy behavior of the GDS model, we would like a description in terms of a
TQFT whose state spaces we can compute relatively easily. The answer comes to us as one of a class of TQFTs, called
$\Z/2$-gauge-gravity theories; these TQFTs are slight generalizations of Dijkgraaf-Witten
theories~\cite{DijkgraafWitten, FreedQuinn}, in which Stiefel-Whitney classes of the underlying manifold can enter
the Lagrangian action. Theories of this sort have also been considered by Kapustin~\cite{KapustinGravity2,
KapustinGravity}, Wen~\cite{WenGravity,WenGravity2}, and Lan-Kong-Wen~\cite{LKWGravity}, though not in this
generality.

As in the construction of Dijkgraaf-Witten theories, we will construct the gauge-gravity theories in two steps:
defining the classical (invertible) theory for unoriented manifolds with a principal $\Z/2$-bundle, then summing
over principal $\Z/2$-bundles to define the quantum theory.

\subsection{Construction of the classical $\Z/2$-gauge-gravity theories}
\label{classical_action}
Let $\Bord_n$ denote the unoriented bordism category in dimension $n$, whose objects are closed $(n-1)$-manifolds
and whose morphisms are diffeomorphism classes of bordisms between them, and, for a topological space $X$, let $\Bord_n(X)$ denote the
bordism category of manifolds together with a map to $X$.
\begin{defn}
A TQFT $Z\colon\Bord_n(X)\to\Vect_\C$ is \term{invertible} if it factors through the subgroupoid
$\Line_\C\inj\Vect_\C$ of complex lines and nonzero homomorphisms.
\end{defn}
This means, for example, that all partition functions are nonzero and all state spaces are one-dimensional.

\begin{thm}
\label{classpartdefnthm}
Let $\beta\in H^n(B\O_n\times B\Z/2;\Z/2)$. Then there is an invertible TQFT $\Zcl_\beta\colon
\Bord_n(B\Z/2)\to\Vect_\C$ of $n$-manifolds equipped with a principal $\Z/2$-bundle, unique up to
isomorphism, such that for any closed $n$-manifold $M$ and principal $\Z/2$-bundle $P\to M$,
\begin{equation}
\label{classpartfn}
	\Zcl_\beta(M, P) = (-1)^{\ang{\beta(M,P), [M]}},
\end{equation}
where $\beta(M, P)$ denotes the pullback of $\beta$ under a map $M\to B\O_n\times B\Z/2$ classifying $TM$ and
$P$.\footnote{The classifying map is unique up to homotopy, so $\beta(M, P)$ does not depend on this choice.}
\end{thm}
\begin{proof}
The assignment~\eqref{classpartfn} is a $\set{\pm 1}$-valued bordism invariant of manifolds eq:wuipped with a
principal $\Z/2$-bundle. Given such a bordism invariant, Yonekura~\cite[\S4.2]{Yon18} constructs an invertible
TQFT valued in $\Line_\C$ whose partition function recovers the bordism invariant, and proves that
it is unique up to isomorphism.
\end{proof}
We call $\Zcl_\beta$ the \term{classical $\Z/2$-gauge-gravity theory} for $\beta$, and call $\beta$ the
\term{Lagrangian} for the theory.
\begin{rem}
The name ``gauge-gravity'' refers to the fact that the Lagrangian $\beta$ can have terms depending both on the
principal $\Z/2$-bundle (a gauge field) and characteristic classes of the underlying manifold (which, due to the
relationship between characteristic classes and curvature, are sometimes called gravitational terms). This idea
also appears for the anomaly TQFTs in~\cite{GM18, STY18}, which are similar to the classical gauge-gravity theories
considered in this paper.
\end{rem}
\begin{rem}
It is also possible to describe $\Zcl_\beta$ homotopically, following the Freed-Hopkins approach to invertible
TQFTs~\cite{FreedHopkins}; we briefly sketch the construction. If a homomorphism of commutative monoids $A\to B$
factors through the subgroup of units $B^\times\inj B$, then it also factors through the group completion $A\to
K(A)$; in a similar way, if a morphism of symmetric monoidal categories $\fC\to\fD$ factors through the Picard
groupoid of units $\fD^\times\inj\fD$, it also factors through the groupoid completion of $\fC$, which is also a
Picard groupoid. The geometric realization of a Picard groupoid $\cat G$ is canonically an infinite loop space, and
its associated spectrum, called the \term{classifying spectrum} of $\cat G$ and denoted $\abs{\cat G}$, is a
\term{stable 1-type}, i.e.\ its only nonzero homotopy groups are $\pi_0\abs{\cat G}$ and $\pi_1\abs{\cat
G}$~\cite{JO12}. The upshot is that an invertible TQFT $\Zcl\colon\Bord_n(B\Z/2)\to\Line_\C$ determines and is
determined up to isomorphism by the homotopy class of the
map
\begin{equation}
	\abs{\Zcl}\colon\abs{\Bord_n(B\Z/2)}\to\abs{\Line_\C}
\end{equation}
it induces on classifying spectra.

If $E$ is a spectrum, let $E\ang{m,n}$ denote the truncation of $E$ to a spectrum with homotopy groups only in
degrees between $m$ and $n$, inclusive. Then there are weak equivalences
\begin{itemize}
	\item $\abs{\Bord_n(B\Z/2)}\simeq (\Sigma\MTO_n\wedge (B\Z/2)_+)\ang{0,1}$~\cite{GMTW,
	Ngu17},\footnote{This fact has been proven or sketched in several additional ways: see also \cite{Aya09, Lur09,
	BM14, AF17, SP17}.} and \item $\abs{\Line_\C}\simeq \Sigma H\C^\times$.
\end{itemize}
Here $\MTO_n$ is a \term{Madsen-Tillmann spectrum}: if $V_n\to B\O_n$ denotes the tautological bundle, $\MTO_n$ is
the Thom spectrum of $-V_n\to B\O_n$.

Therefore an isomorphism class of invertible $n$-dimensional TQFTs for manifolds with a principal $\Z/2$-bundle is
determined by an element of
\begin{equation}
	[(\Sigma\MTO_n\wedge (B\Z/2)_+)\ang{0,1}, \Sigma H\C^\times] \cong H^0(\MTO_n\wedge (B\Z/2)_+;
	\C^\times),
\end{equation}
and $\beta\in H^n(B\O_n\times B\Z/2;\Z/2)$ yields such an element through the mod 2 Thom isomorphism followed by
the map induced on cohomology by $\Z/2 \cong\set{\pm 1}\inj\C^\times$. Thus it defines an invertible TQFT
$(\Zcl_\beta)'$ up to isomorphism. Tracing through the Pontrjagin-Thom construction, one can prove that its
partition functions agree with those in~\eqref{classpartfn}, and hence by Yonekura's uniqueness result~\cite[Theorem
4.4]{Yon18}, $(\Zcl_\beta)'\cong\Zcl_\beta$.

This approach readily generalizes to extended invertible TQFTs, as in~\cite{SP17}, and the classical gauge-gravity
TQFTs can be realized as fully extended TQFTs valued in $n$-algebras, as in~\cite[\S8]{FHLT}, or $n$-vector
spaces, using the calculation of the classifying spectrum of the $n$-category of $n$-vector spaces
in~\cite[\S7.4]{SP17}.
\end{rem}
The partition functions of the classical gauge-gravity TQFT for $\beta$ resemble those of classical
Dijkgraaf-Witten theory~\cite{DijkgraafWitten, FreedQuinn} for the gauge group $\Z/2$, though the Lagrangians of the
former can also contain Stiefel-Whitney classes. If $\beta$ factors through the inclusion $H^n(B\Z/2;\Z/2)\inj
H^n(B\O_n\times B\Z/2;\Z/2)$, then $\Zcl_\beta$ is isomorphic to a classical $\Z/2$-Dijkgraaf-Witten theory.

If $\gamma\in H^n(B\Z/2;\R/\Z)$, we let $\DWcl_\gamma$ denote classical $\Z/2$-Dijkgraaf-Witten theory with
Lagrangian $\gamma$.
\begin{prop}
\label{is_classical_DW}
Let $f\colon\Z/2\inj\R/\Z$ denote the map sending $1\mapsto 1/2$, as well as the map $f\colon H^*(X;\Z/2)\to
H^*(X;\R/\Z)$ it induces on cohomology. Suppose $\beta$ contains no Stiefel-Whitney terms, i.e.\ $\beta$ factors
through $H^n(B\Z/2;\Z/2)\inj H^n(B\O_n\times B\Z/2;\Z/2)$. Then, as TQFTs of oriented manifolds equipped with
principal $\Z/2$-bundles, $\Zcl_\beta\cong\DWcl_{f(\beta)}$.
\end{prop}
\begin{proof}
Let $M$ be a closed, oriented $n$-manifold, $P\to M$ be a principal $\Z/2$-bundle, and $\beta$ be as in the
proposition statement. Let $\phi\colon M\to B\Z/2$ be a classifying map for $P$.  Let $[M]_\Z$, resp.\
$[M]_{\Z/2}$, denote the fundamental class of $M$ in integral, resp.\ $\Z/2$, homology.

The partition function of classical $\Z/2$-Dijkgraaf-Witten theory with Lagrangian $f(\beta)$ is $\DWcl_{f(\beta)}(M,
P) = e^{i\pi\ang{(\phi^*(f(\beta)), [M]_\Z}}$~\cite[Theorem 1.7]{FreedQuinn}. Naturality of the cap product under change of coefficients implies $f(\ang{x, [M]_{\Z/2}}) = \ang{f(x), [M]_\Z}$ for any $x\in H^n(M;\Z/2)$, and naturality of the
change-of-coefficients map on cohomology implies that $\phi^*(f(\beta)) = f(\phi^*(\beta))$, so
$f(\ang{\phi^*\beta, [M]_{\Z/2}}) = \ang{\phi^*(f(\beta)), [M]_\Z}$. If $a\in\Z/2$, $(-1)^a = e^{i\pi f(a)}$, so
\begin{equation}
	Z_\beta(M, P) = (-1)^{\ang{\phi^*\beta, [M]_{\Z/2}}} = e^{i\pi\ang{\phi^*(f(\beta)), [M]_\Z}} =
	\DWcl_{f(\beta)}(M, P).
\end{equation}
Since the partition functions for these theories are identical, then by~\cite[Theorem 4.4]{Yon18},
$Z_\beta\cong\DWcl_{f(\beta)}$.
\end{proof}
\begin{lem}
Let $\gamma\in H^n(B\O_n\times B\Z/2;\Z/2)$ be a cohomology class which vanishes when pulled back to all closed
$n$-manifolds via a classifying map for the tangent bundle and any principal $\Z/2$-bundle. Then,
$\Zcl_\beta\cong\Zcl_{\beta+\gamma}$.
\end{lem}
\begin{proof}
By~\eqref{classpartfn}, $Z_\beta(M) = Z_{\beta+\gamma}(M)$ for all closed $n$-manifolds $M$ with a principal
$\Z/2$-bundle. We have seen that invertible TQFTs of manifolds with a principal $\Z/2$-bundle are determined up to
isomorphism by their partition functions, so $Z_\beta\cong Z_{\beta+\gamma}$.
\end{proof}
For example, in dimension 3, $w_1^2\alpha = w_2\alpha$ on all 3-manifolds, so $\Zcl_{w_1^2\alpha}\cong
\Zcl_{w_2\alpha}$.

\begin{cor}
\label{all_TDW_arise}
If $n$ is odd, every classical $\Z/2$-Dijkgraaf-Witten theory is isomorphic to $\Zcl_\beta$ for some $\beta\in
H^n(B\Z/2;\Z/2)\inj H^n(B\O_n\times B\Z/2;\Z/2)$.
\end{cor}
\begin{proof}
When $n$ is odd, the map $f\colon H^n(B\Z/2;\Z/2)\to H^n(B\Z/2;\R/\Z)$ is surjective; then the result follows from
\cref{is_classical_DW}.
\end{proof}
\subsection{Discussion of the quantum theories}
\label{constquant}

We construct the quantum theory $Z_\beta$ using the finite path integral approach of~\cite[\S3]{FHLT}; see
also~\cite{Mor15, Tro16} for a more detailed account and~\cite{SW18} for a related construction. Let $\Gpd$ denote
the category of spans of finite groupoids: the objects of $\Gpd$ are finite groupoids, and a morphism from $X_1$ to
$X_2$ is data of a finite groupoid $Y$ and functors $p_1\colon Y\to X_1$ and $p_2\colon Y\to X_2$, considered up to equivalence of $(Y, p_1, p_2)$. Let
$\Gpd(\Vect_\C)$ denote the category whose objects are pairs $(X, V)$, where $X$ is a groupoid and $V\to X$ is a
complex vector bundle,\footnote{A (complex) vector bundle over a groupoid $\cG$, denoted $V\to\cG$, is a functor
$V\colon \cG\to\Vect_\C$, and its space of sections is $\varinjlim L$. We will always assume these vector bundles
are finite-dimensional, meaning they factor through the full subcategory of finite-dimensional vector spaces.} and
whose morphisms are equivalence classes of spans
\begin{equation}
\label{spaninlocsys}
\begin{gathered}
\xymatrix@dr{
	Y\ar[r]^{p_2}\ar[d]_{p_1} & X_2\\
	X_1
}
\end{gathered}
\end{equation}
together with data of vector bundles $V_i\to X_i$ and $W\to Y$ and morphisms $\phi_i\colon p_i^*V_i\to W$ for $i
=1,2$. For any $y\in Y$, this morphism determines a linear map $\vp(y)\colon V_1(p_1(y))\to V_2(p_2(y))$ by a
push-pull construction. Disjoint union of groupoids defines a symmetric monoidal structure on $\Gpd(\Vect_\C)$.

We next define the ``quantization'' functor $\Sigma\colon\Gpd(\Vect_\C)\to\Vect_\C$, which on to an object assigns
\begin{equation}
	\label{what_is_quantization}
	\Sigma\colon (X, V)\mapsto \Gamma(V)\coloneqq \varinjlim_{x\in X} V(x),
\end{equation}
i.e.\ regard $V$ as a $\Vect_\C$-valued diagram indexed by the category $X$, and take the colimit of this diagram.
Given a morphism $(Y,W,\phi_1,\phi_2)$ as above, the maps $\vp(y)$ for $y\in Y$ pass to the colimit to define a map
\begin{equation}
	\widetilde\vp\colon\pi_0 Y\to\Hom(\Gamma(X_1, V_1), \Gamma(X_2, V_2)).
\end{equation}
Then, $\Sigma$ assigns to this morphism the linear map
\begin{equation}
	\Sigma(Y, W)\coloneqq \sum_{[y]\in\pi_0Y}\frac{\widetilde\vp(y)}{\abs{\Aut(y)}}\in\Hom(\Gamma(X_1, V_1),
	\Gamma(X_2, V_2)).
\end{equation}
This functor is symmetric monoidal.

Given a TQFT $\Zcl\colon\Bord_n(B\Z/2)\to\Vect_\C$, the functor
$F_{\Zcl}\colon\Bord_n\to\Gpd(\Vect_\C)$ sending
\begin{equation}
	F_{\Zcl}\colon M\mapsto \paren{\Bun_{\Z/2}(M), P\mapsto \Zcl(M, P)}
\end{equation}
is also symmetric monoidal, and therefore the composition
\begin{equation}
\label{quanteqn}
\xymatrix{
	Z\colon\Bord_n\ar[r]^-{F_{\Zcl}} & \Gpd(\Vect_\C)\ar[r]^-\Sigma & \Vect_\C
}\end{equation}
is symmetric monoidal, i.e.\ a (nonextended) TQFT of unoriented manifolds.
\begin{defn}
\label{quantdefn}
Given a TQFT $\Zcl\colon\Bord_n(B\Z/2)\to\Vect_\C$, the TQFT $Z$ in~\eqref{quanteqn} above is called the
\term{quantum theory} associated to $\Zcl$. In particular, we denote the quantum theory associated to $\Zcl_\beta$
by $Z_\beta$, and call it the \term{(quantum) gauge-gravity theory} for $\beta$. In this case we call $\beta$ the
\term{Lagrangian} of the theory.
\end{defn}

\begin{prop}\hfill
\label{quantchar}
\begin{enumerate}
	\item\label{quantpartfn} Let $M$ be a closed $n$-manifold. Then, the partition function $Z_\beta(M)$ is
	\begin{equation}
		Z_\beta(M) = \sum_{[P]\in\pi_0\Bun_{\Z/2}(M)} \frac{(-1)^{\ang{\beta(P), [M]}}}{\abs{\Aut(P)}}.
	\end{equation}
	\item\label{quantstatespace} Let $N$ be a closed $(n-1)$-manifold. Then, define a line bundle
	$L_\beta\to\Bun_{\Z/2}(N)$ which
	\begin{itemize}
		\item assigns $\C$ to every object, and
		\item assigns to an automorphism $\phi\in\Aut(P)$ multiplication by $\Zcl_\beta(S^1\times N, P_\phi)$.
	\end{itemize}
	Then the state space of $N$ is $Z_\beta(N)\cong\Gamma(L_\beta)$.
\end{enumerate}
\end{prop}
Here $P_\phi\to S^1\times N$ denotes the \term{mapping torus} of $\phi$, i.e.\ the quotient of $[0,1]\times P$ by
$(0,x)\sim(1,\phi(x))$.
\begin{proof}
The proof boils down to figuring out what~\eqref{what_is_quantization} means in this context.

The partition function for $M$ is defined to be $Z_\beta(M\colon \varnothing\to\varnothing)$; since
$Z_\beta\coloneqq \Sigma\circ F_{\Zcl_\beta}$, we look at each functor in turn.

First, $F_{\Zcl_\beta}$ assigns to the bordism $M\colon\varnothing\to\varnothing$ a span
\begin{equation}
\begin{gathered}
\xymatrix@dr@R=1.3cm@C=1.3cm{
	(L\to\Bun_{\Z/2}(M))\ar[r]\ar[d] & (\underline\C\to\Bun_{\Z/2}(\varnothing)),\\
	(\underline\C\to\Bun_{\Z/2}(\varnothing))
}
\end{gathered}
\end{equation}
for some line bundle $L\to\Bun_{\Z/2}(M)$ such that for any $P\in\Bun_{\Z/2}(M)$, the induced map
$\vp(P)\colon\C\to\C$ is multiplication by the classical partition function $\Zcl_{\beta}(M, P)$. Therefore
\begin{equation}
	Z_\beta(M)\coloneqq \Sigma(L, \Bun_{\Z/2}(M)) =
	\sum_{[P]\in\pi_0\Bun_{\Z/2}(M)}\frac{\widetilde\vp(P)}{\abs{\Aut(P)}} =
	\sum_{[P]\in\pi_0\Bun_{\Z/2}(M)} \frac{\Zcl_{\beta}(M, P)}{\abs{\Aut(P)}}.
\end{equation}
Together with~\eqref{classpartfn}, this proves part~\eqref{quantpartfn}.

We address part~\eqref{quantstatespace} in a similar way. $F_{\Zcl_\beta}$ sends $N$ to a line bundle
$L_N\to\Bun_{\Z/2}(N)$, which to a principal $\Z/2$-bundle $P\to N$ assigns the complex line $\Zcl_\beta(N,P)$. For
a morphism $\phi\in\Aut(P)$, we first realize $\phi$ as a bordism: let $\Cyl^\phi(P)\to[0,1]\times N$ denote the
\term{mapping cylinder} of $\phi$, i.e.\ the space $P\times[0,1]\to N\times[0,1]$, interpreted as a bordism in
which $P$ is glued by the identity at $0$ and by $\phi$ at $1$. Then,
\begin{equation}
\label{preMCGtrace}
	L_N(\phi) = \Zcl_\beta([0,1]\times N, \Cyl^\phi(P))\colon\Zcl_\beta(N, P)\to\Zcl_\beta(N, P).
\end{equation}
General facts about TQFT imply that for any TQFT $Z\colon\Bord_n(B\Z/2)\to\Vect_\C$,
\begin{equation}
\label{MCGtrace}
	Z(S^1\times N, P_\phi) = \tr(Z([0,1\times N, \Cyl^\phi(P))\colon Z(N, P)\to Z(N, P)).
\end{equation}
This simplifies~\eqref{preMCGtrace} to
\begin{equation}
	L_N(\phi) = (\text{multiplication by }\Zcl_{\beta}(S^1\times N, P_\phi))\colon\Zcl_\beta(N, P)\to\Zcl_\beta(N,
	P).
\end{equation}
Thus $L_N\to\Bun_{\Z/2}(N)$ is isomorphic to the line bundle $L_\beta$ in the statement of the proposition, and
\begin{equation}
	Z_\beta(N) = \Sigma(\Bun_{\Z/2}(N), L_N) = \varinjlim L_N \cong \varinjlim L_\beta = \Gamma(L_\beta). \qedhere
\end{equation}
\end{proof}

The finite path integral approach to defining the quantum gauge-gravity theories means a few of their basic
properties are formal corollaries of their counterparts in the classical case, because an isomorphism of classical
theories determines an isomorphism of quantum theories.

\begin{cor}
\label{Wuisom}
Let $\gamma\in H^n(B\O_n\times B\Z/2;\Z/2)$ be a cohomology class which vanishes when pulled back to all closed
$n$-manifolds via a classifying map for the tangent bundle and any principal $\Z/2$-bundle. Then, $Z_\beta\cong
Z_{\beta+\gamma}$.
\end{cor}
\begin{cor}
\label{quantumTDW}
Suppose $\beta$ contains no Stiefel-Whitney terms (in the sense of \cref{is_classical_DW}). Then,
$Z_\beta\cong\DDW_\beta$, the quantum $\Z/2$-Dijkgraaf-Witten theory with Lagrangian $\beta$.
\end{cor}
\begin{cor}
\label{all_quantum_TDW}
If $n$ is odd, every quantum $\Z/2$-Dijkgraaf-Witten theory is isomorphic to $Z_\beta$ for some $\beta\in
H^n(B\Z/2;\Z/2)\inj H^n(B\O_n\times B\Z/2;\Z/2)$.
\end{cor}
There is a new phenomenon at this level, however: one can produce $\beta$ and $\beta'$ whose quantum theories are
isomorphic, but whose classical theories are not.
\begin{defn}
Let $\beta\in H^n(B\O_n\times B\Z/2;\Z/2)$, so that there are coefficients $\gamma_1,\dotsc,\gamma_n\in
H^*(B\O_n;\Z/2)$ such that
\begin{equation}
	\beta = \gamma_n\alpha^n + \gamma_{n-1}\alpha^{n-1} + \dotsb + \gamma_1\alpha + \gamma_0,
\end{equation}
where $\alpha\in H^1(B\Z/2;\Z/2)$ is the generator. If $w_1\in H^1(B\O_n;\Z/2)$ denotes the first Stiefel-Whitney
class, we call
\begin{equation}
	\beta_{w_1}\coloneqq \gamma_n(\alpha +w_1)^n + \gamma_{n-1}(\alpha +w_1)^{n-1} + \dotsb + \gamma_1(\alpha +
	w_1) + \gamma_0 \in H^n(B\O_n\times B\Z/2;\Z/2)
\end{equation}
the \term{orientation-twisting} of $\beta$.
\end{defn}
\begin{prop}
\label{w_1isom}
Let $\beta_{w_1}$ be the orientation-twisting of $\beta$. Then, $Z_\beta\cong Z_{\beta_{w_1}}$.
\end{prop}
The idea is that replacing $\beta$ with $\beta_{w_1}$ corresponds to tensoring with the orientation bundle, an
involution on the space of fields. Since we are summing over the fields, this does not change the path integral.
\begin{defn}
We define a tensor product of principal $\Z/2$-bundles induced from the tensor product of real line bundles. Given
two principal $\Z/2$-bundles $P_1,P_2\to M$, define a real line bundle $L(P_i)\to M$ for $i = 1,2$ by
$L(P_i)\coloneqq P_i\times_{\Z/2}\underline\R$, where $\Z/2$ acts on $\underline\R$ as $\set{\pm 1}$. The Euclidean
metric on $\underline\R$ induces Euclidean metrics on $L(P_1)$ and $L(P_2)$, hence also on $L(P_1)\otimes L(P_2)$;
we define the \term{tensor product} of $P_1$ and $P_2$, denoted $P_1\otimes P_2\to M$, to be the unit sphere bundle
in $L(P_1)\otimes L(P_2)$, which is a principal $\Z/2$-bundle on $M$.
\end{defn}
The characteristic class of $P\otimes Q$ is $\alpha(P\otimes Q) = \alpha(P) + \alpha(Q)$.
%

On any manifold $M$, there is a canonical principal $\Z/2$-bundle $\fo_M$, called the \term{orientation bundle},
whose fiber at $x\in M$ is the $\Z/2$-torsor of orientations at $x$. Its characteristic class is $\alpha(\fo_M)
= w_1(M)$.
\begin{proof}[Proof of \cref{w_1isom}]
Let $\PM_n$ denote the subcategory of $\Gpd(\Vect_\C)$ whose objects are vector bundles over groupoids of the form
$\Bun_{\Z/2}(N)$ for some closed $(n-1)$-manifold $N$ and whose morphisms are induced from the spans
\begin{equation}
\label{pushpull}
\begin{gathered}\xymatrix@dr{
	\Bun_{\Z/2}(M)\ar[r]\ar[d] & \Bun_{\Z/2}(N_2),\\
	\Bun_{\Z/2}(N_1)
}\end{gathered}
\end{equation}
where $M$ is a bordism between $N_1$ and $N_2$. For any $\beta$, $F_{\Zcl_\beta}$ lands in $\PM_n$. To simplify
notation, we will let $F_\beta\coloneqq F_{\Zcl_\beta}$.

If $M$ is a bordism between $N_1$ and $N_2$, $(\fo_M)|_{N_i} = \fo_{N_i}$. Thus the automorphism
$\bl\otimes\fo_Y\colon\Bun_{\Z/2}(Y)\to\Bun_{\Z/2}(Y)$ induces an automorphism $\Phi\colon\PM_n\to\PM_n$ as
follows.
\begin{itemize}
	\item An object of $\PM_n$ is a functor $F\colon \Bun_{\Z/2}(N)\to\Vect_\C$ for some $(n-1)$-manifold $N$. Let
	$\Phi(F)$ be $F\circ(\bl\otimes\fo_N)\colon\Bun_{\Z/2}(N)\to\Bun_{\Z/2}(N)\to\Vect_\C$.
	\item A morphism $F_1\to F_2$ of $\PM_n$ is a push-pull map induced from a span as in~\eqref{pushpull}. Since
	$(\fo_M)|_{N_i} = \fo_{N_i}$, the arrows in~\eqref{pushpull} intertwine the actions of $\bl\otimes\fo_M$ and
	$\bl\otimes\fo_{N_i}$, so this span induces a morphism $\Phi(F_1)\to\Phi(F_2)$ as desired.
\end{itemize}
Thus we may consider the diagram
\begin{equation}
\begin{gathered}
\xymatrix{
	\Bord_n\ar[r]^-{F_\beta}\ar[dr]_-{F_{\beta_{w_1}}} &\PM_n\ar[r]^-\Sigma\ar[d]^\Phi & \Vect_\C\\
	& \PM_n\ar[ur]_-\Sigma,
}
\end{gathered}
\end{equation}
where the composition along the top is $Z_\beta$ and the composition along the bottom is $Z_{\beta_{w_1}}$.

It suffices to prove this diagram commutes up to natural isomorphism, which means checking its two triangles.
\begin{itemize}
	\item The left triangle commutes (up to natural isomorphism) by design, since $\alpha(P\otimes\fo_M) =
	\alpha(P) + w_1(M)$ and in $\beta_{w_1}$, we have replaced $\alpha$ with $\alpha + w_1$.
	\item The right triangle commutes because $\Sigma$ takes a diagram and evaluates its colimit, and an
	automorphism of the indexing category does not change the value of the colimit. Hence $\Sigma(S)$ and
	$(\Sigma\circ\Phi)(S)$ are isomorphic for any object $S$, and since $\Phi$ is compatible with morphisms in
	$\PM_n$, $\Sigma$ and $\Sigma\circ\Phi$ also agree on morphisms. \qedhere
\end{itemize}
\end{proof}

\begin{exm}
\label{notDWclassical}
The orientation twisting of $\alpha^2$ is $\alpha^2 + w_1^2$. The classical theories $\Zcl_{\alpha^2}$ and
$\Zcl_{\alpha^2 +w_1^2}$ are nonisomorphic; for example, they disagree on $\RP^2$ with the trivial principal
$\Z/2$-bundle. But by \cref{w_1isom}, their quantum theories are isomorphic.
\end{exm}
\begin{rem}
Lu-Vishwanath~\cite{LV16} observe a similar phenomenon in the physics of topological phases enriched by a global
$\Z/2$-symmetry, in which distinct phases become equivalent after gauging the $\Z/2$ symmetry.
\end{rem}

\section{Low-energy limits}
	\label{low-energy_limits}
	In this section, we return to the lattice, and investigate the spaces of ground states of the toric code and GDS
models on closed $(n-1)$-manifolds. In both cases, we find a TQFT $Z$ whose state space on $M$ is isomorphic to the
space of ground states of the lattice model on $M$.
\begin{defn}
\label{LEEFT}
Consider a lattice model which to all closed $(n-1)$-manifolds $M$ together with some kind of lattice $\Pi$ (e.g.\
a triangulation or a CW structure) associates a complex Hilbert space $\cH_{M,\Pi}$ and a self-adjoint operator
$H_{M,\Pi}\colon\cH_{M,\Pi}\to\cH_{M,\Pi}$ (respectively the state space and the Hamiltonian).  In this setting,
elements of $\ker(H_{M,\Pi})$ are called \term{ground states}.

Let $Z\colon\Bord_n\to\Vect_\C$ be a TQFT. We say that $Z$ \term{captures the ground states} of the lattice model
if for all closed $(n-1)$-manifolds $M$ with a lattice $\Pi$, $Z(M)\cong\ker(H_{M,\Pi})$.
\end{defn}
\begin{rem}
When $Z$ captures the ground states of a lattice model, it is believed to correspond to the physics notion of the
low-energy effective theory of the model. The existence of such a low-energy TQFT for certain lattice models,
called topological phases, is predicted by physics,\footnote{One should allow TQFTs tensored with an invertible,
non-topological theory, as in~\cite[\S5.4]{FreedHopkins}. The TQFTs we find in this paper are topological, so this
distinction will not matter here.} and the low-energy TQFT is expected to determine the lattice model up to some
physically meaningful notion of equivalence; this correspondence is discussed in~\cite{FreedHopkins, Gaiotto,
RowellWang, FT18}.

However, there is much left to understand, especially at a mathematical level of rigor. \Cref{LEEFT} is structured
to make \cref{TC_deriv,mainthm} easier to state; we do not intend for it to be a mathematical definition of the
physical notion of the low-energy effective theory of a lattice model. Providing such a mathematical definition is
a major open question; as is, \cref{LEEFT} fails to address uniqueness (as shown in \cref{another_LE}) and
existence (due to fracton phases; see, e.g.\ \cite{BLT11, Haa11, Yos13}).
\end{rem}

	\subsection{Review for the toric code}
		\label{LETC}
As a warmup, before tackling the GDS model, we determine a TQFT which captures the ground states of the toric code.
Neither the answer nor this perspective on it are new.
\begin{thm}
\label{TC_deriv}
Let $\DW\colon\Bord_n\to\Vect_\C$ denote the $\Z/2$-Dijkgraaf-Witten theory with Lagrangian equal to $0$. Then
$\DW$ captures the ground states of the toric code.
\end{thm}
\begin{rem}
This is not a new result. Because researchers consider different formulations of the toric code, there are some
analogues of \cref{TC_deriv} in the literature for different classes of toric code models, e.g.\ in~\cite{Kitaev,
BK12, Cha14}. Though these results do not cover \cref{TC_deriv} in the case $n > 3$, it and its proof were
certainly known before this paper.
\end{rem}
We can use the fact that the vertex and face operators commute to simplify our analysis of the Hamiltonian.
\begin{lem}
\label{A+Blinalg}
Let $V$ be a vector space over a field $k$, and let $\Phi = \sum_{i=1}^m \phi_i$ be a finite sum of commuting
projections $\phi_i\in\End_k(V)$. Then, $\ker(\Phi) = \bigcap_{i=1}^m \ker(\phi_i)$.
\end{lem}
\begin{proof}
By induction, it suffices to consider $m = 2$, so $\Phi = \phi_1 + \phi_2$. Clearly
$\ker(\phi_1)\cap\ker(\phi_2)\subset\ker(\Phi)$, so assume $\Phi x = 0$ for some $x\in V$. Thus $\phi_1x = -\phi_2
x$, so $\phi_1x = \phi_1^2x = -\phi_1\phi_2x = -\phi_2(\phi_1x)$, so $\phi_1x$ is an eigenvector for $\phi_2$ with
eigenvalue $-1$. This means $\phi_2^2(\phi_1x) = (-1)^2\phi_1x = \phi_1x$, and since $\phi_2$ is a projection,
$\phi_2^2(\phi_1x) = \phi_2\phi_1x = -\phi_1x$, forcing $\phi_1x = 0$. Since $\phi_2 = A - \phi_1$, then $\phi_2x =
0$ as well.
\end{proof}
Our proof of \cref{TC_deriv} will be slightly more complicated than necessary. This is so that it follows the same
line of argument as the proof for the GDS model in \S\ref{GDS_deriv}. We hope that presenting the simpler example
first makes the GDS example easier to understand.
\begin{proof}[Proof of \cref{TC_deriv}]
Let $M$ be a closed manifold with a CW structure $\Xi$. As before, we will write $(P,\xi)$ for an object of
$\Bun_{\Z/2}(M^1, M^0)$, meaning that $P\to M^1$ is a principal $\Z/2$-bundle and $\xi\colon M^0\to P|_{M^0}$ is a
trivialization of $P$ over $M^0$.

By \cref{A+Blinalg}, the ground states of the toric code for $M$ are those functions $\psi$ on $\Bun_{\Z/2}(M^1,
M^0)$ such that $H_v\psi = 0$ for all $0$-cells $v$ and $H_f\psi = 0$ for all $2$-cells $f$. 

Let $f$ be a $2$-cell. Then, $H_f\psi = 0$ if and only if $B_f\psi = \psi$, or for all $(P,\xi)\in\Bun_{\Z/2}(M^1,
M^0)$, $(-1)^{\Hol_P(f)}\psi(P,\xi) = \psi(P,\xi)$. That is, either $\psi(P,\xi) = 0$ or $\Hol_P(f) = 0$, so $\psi$
must vanish on all principal $\Z/2$-bundles with nontrivial holonomy around $\partial f$. Hence if
$\psi\in\ker(H_f)$ for all $2$-cells $f$, it can only be nonzero on the principal $\Z/2$-bundles with no holonomy
around the boundary of any $2$-cell, which are exactly the principal $\Z/2$-bundles which extend to $M^2$, hence to
all of $M$, and such an extension is necessarily unique. That is, $\bigcap_f \ker(H_f)$ is the space of functions
on $\Bun_{\Z/2}(M, M^0)$.

Let $\cA\coloneqq C_\Xi^0(M;\Z/2)$ denote the group of cellular $0$-cochains. We will describe the ground states of
the toric code for $M$ as invariant sections of an $\cA$-equivariant line bundle on $\Bun_{\Z/2}(M, M^0)$, then
take the quotient by $\cA$. For $v\in\Delta^{0}(M)$, let $\delta_v\in\cA$ be the function equal to $1$ on $v$ and
$0$ elsewhere. Then, $\cA$ has a presentation by the following generators and relations:
\begin{equation}
	\cA \cong \ang{\delta_v\text{ for all } v\in V\mid \delta_v^2, [\delta_v, \delta_w]},
\end{equation}
so an $\cA$-action is the same data as commuting involutions associated to each $\delta_v$. For example, $\cA$ acts
on the (discrete) groupoid $\Bun_{\Z/2}(M, M^0)$ through the commuting involutions
\begin{equation}
	\delta_v\colon (P,\xi)\mapsto (P, (w\mapsto \xi(w) + \delta_v(w))).
\end{equation}
Consider the trivial line bundle $\underline\C\to\Bun_{\Z/2}(M, M^0)$ and give it the trivial $\cA$-action. We can
identify sections of $\underline\C$ with functions on $\Bun_{\Z/2}(M, M^0)$, and the $\cA$-actions match; in
particular, if $\psi\in\Gamma(\underline\C)$ and $v$ is a $0$-cell, then $\delta_v\cdot \psi = A_v\psi$. Therefore
$\psi$ is invariant under the $\cA$-action if and only if $A_v\psi = \psi$ for all $v$, i.e.\ $H_v\psi = 0$ for all
$v$. That is, the space of ground states is the space of $\cA$-invariant sections of $\underline\C\to\Bun_{\Z/2}(M,
M^0)$.

The $\cA$-equivariant line bundle $\underline\C\to\Bun_{\Z/2}(M, M^0)$ descends to a nonequivariant line bundle on
the groupoid quotient $\Bun_{\Z/2}(M, M^0)/\cA$; since we began with the trivial $\cA$-action, this will also be
a trivial line bundle. Therefore it suffices to identify the quotient.
\begin{lem}
\label{quotientBun}
The map $\Bun_{\Z/2}(M, M^0)/\cA\to\Bun_{\Z/2}(M)$ which forgets the trivialization is an equivalence of groupoids.
Given $(P,\xi)\in\Bun_{\Z/2}(M, M^0)$ and $\phi\in\Aut(P)$, action by
\begin{equation}
        t_\phi\coloneqq \sum_{\substack{v\in \Delta^0(M)\\\phi|_v\text{\rm{} nontrivial}}} \delta_v \in\cA
\end{equation}
on $(P,\xi)$ passes to $\phi$ in the quotient.
\end{lem}
\begin{proof}
$\Bun_{\Z/2}(M, M^0)$ is a discrete groupoid, so we just have to determine the stabilizer subgroup for the
$\cA$-action. An automorphism $\phi$ of $P$ switches the trivializations wherever $\phi$ is nontrivial, so defines
an isomorphism $(P,\xi)\overset\cong\to (P, t_\phi\cdot \xi)$. To check these are the only isomorphisms that occur,
suppose $(P,\xi)\cong (P,t\cdot\xi)$ for some $t\in\cA$. Since the function $\spin_{(P,\xi)}$ is an isomorphism
invariant of $(P,\xi)\in\Bun_{\Z/2}(M, M^0)$, $t$ must be the sum of $\delta_v$ as $v$ ranges over a set $S$ of
$0$-cells such that every $1$-cell of $M$ bounds an even number of $0$-cells in $S$. Thus for any connected
component $M_0$ of $M$, $S$ includes either all $0$-cells of $M_0$ or none, so $t$ is realized by some $t_\phi$.
\end{proof}
Therefore the space of ground states on $M$ is the space of sections of $\underline\C\to\Bun_{\Z/2}(M)$, i.e.\ the
space of functions on $\Bun_{\Z/2}(M)$, which is what $\DW$ assigns to $M$.
\end{proof}

	\subsection{Derivation of the generalized double semion Lagrangian}
		\label{GDS_deriv}
We now answer the main question of this paper: identifying a TQFT whose state spaces are isomorphic to the spaces
of ground states of the GDS model.
\begin{defn}
Fix a dimension $n$ and let $\beta\in H^n(B\O_n\times B\Z/2;\Z/2)$ denote the degree-$n$ part of
$w\alpha/(1+\alpha)$, where $w$ is the total Stiefel-Whitney class and $\alpha$ is the generator of
$H^1(B\Z/2;\Z/2)$. We let $\DS\colon\Bord_n\to\Vect_\C$ denote the quantum gauge-gravity theory $Z_\beta$ from
\cref{quantdefn}; the dimension $n$ will be clear from context when needed.
\end{defn}
Our goal in this section is to prove the following.
\begin{thm}
\label{mainthm}
The TQFT $\DS$ captures the ground states of the GDS model.
\end{thm}
Let $M$ be a closed $(n-1)$-manifold with a smooth triangulation $\Pi$; as in \S\ref{origGDS}, we assume the
$0$-clopen star of any vertex is contractible. We will prove \cref{mainthm} by identifying the ground states of the
GDS model on $M$ with the space of sections of a line bundle $\LDS\to\Bun_{\Z/2}(M)$ defined below.
\Cref{quantchar} identifies $\DS(M)$ with the sections of another line bundle $L_\beta\to\Bun_{\Z/2}(M)$, and we
will show that $\LDS\cong L_\beta$.

\subsubsection{Defining $\LDS\to\Bun_{\Z/2}(M)$}
The commutativity relations for the operators in the GDS model are more complicated than those for the toric code,
but we can still understand the spaces of ground states in terms of the vertex and face operators.
\begin{lem}
\label{GDSrelns}
With $V$ as in \cref{A+Blinalg}, let $\phi_i, \psi_j\in\End_k(V)$ and suppose
\begin{equation}
	H = \underbrace{\sum_{i=1}^\ell \phi_i}_\Phi + \underbrace{\sum_{j=1}^m \psi_j}_\Psi,
\end{equation}
such that for all $i$ and $j$,
\begin{enumerate}
	\item $\phi_i$ and $\psi_j$ are projections,
	\item $[\phi_i,\phi_j] = 0$,
	\item $[\phi_i,\psi_j] = 0$,
	\item for any $x\in\ker(\Phi)$, $[\psi_i,\psi_j]x = 0$.
\end{enumerate}
Then,
\begin{equation}
	\ker(H) = \bigcap_{j=1}^m\ker(\psi_j\colon \ker(\Phi)\to\ker(\Phi)).
\end{equation}
\end{lem}
\begin{proof}
\Cref{A+Blinalg} tells us $\ker(H) = \ker(\Phi)\cap\ker(\Psi)$, so it suffices to restrict to $\ker(\Phi)$. Since
$\phi_i$ and $\psi_j$ commute, then $\psi_j(\ker \Phi)\subset\ker \Phi$ for each $j$, so we may consider $\psi_j$
as an operator on $\ker(\Phi)$. Restricted to this subspace, $[\psi_i,\psi_j] = 0$, so we apply \cref{A+Blinalg}
again to conclude.
\end{proof}
The upshot is that for a Hamiltonian whose smallest eigenvalue is $0$ and which is a sum of vertex and face
operators satisfying the commutativity conditions in \cref{GDSrelns}, the space of ground states can be computed by
finding the $f\in\cH$ with $\phi_if = 0$ for all $i$, then taking the subspace of those such that $\psi_jf = 0$ for
all $j$.

By \cref{commuting_descent,DS_proj}, the vertex and face operators for the GDS model satisfy the commutation
relations in \cref{GDSrelns}, where the $\phi_i$ are the face operators and the $\psi_j$ are the vertex operators,
so we will use this method to find the space of ground states.

The first part of the derivation is to determine $\bigcap_f \ker(H_f)$. The $H_f$ operators in the GDS model are
the same as in the toric code, so the derivation proceeds as for the toric code (the first part of the proof of
\cref{TC_deriv}) to produce the space of functions on $\Bun_{\Z/2}(M, M^0)$.

Next, we will use the vertex operators to define $\LDS\to\Bun_{\Z/2}(M)$ and characterize the ground states on $M$
as its space of sections. Specifically, letting $\cA\coloneqq C^0_\Pi(M;\Z/2)$ as in the previous section, we will
describe an $\cA$-equivariant line bundle on $\Bun_{\Z/2}(M, M^0)$ whose invariant sections are the ground states,
then let $\LDS\to\Bun_{\Z/2}(M)$ denote the induced bundle on the quotient.
\begin{defn}
First, we define the $\cA$-equivariant line bundle $\LDS'\to\Bun_{\Z/2}(M, M^0)$. Begin with the trivial
(nonequivariant) line bundle $\underline\C\to\Bun_{\Z/2}(M, M^0)$, and give it an $\cA$-action as follows: if
$(P,\xi)\in\Bun_{\Z/2}(M,M^0)$ and $z\in\C$, let
\begin{equation}
\label{GDSAact}
	\delta_v\colon ((P,\xi), z)\mapsto (\delta_v\cdot (P,\xi), \sigma(v, (P,\xi))z),
\end{equation}
where $\sigma(v,(P,\xi))$ is the GDS sign from~\eqref{GDS_sign}. By \cref{commuting_descent,DS_proj}, the actions
of $\delta_{v_1}$ and $\delta_{v_2}$ on $\underline\C$ commute for $0$-cells $v_1$ and $v_2$, so~\eqref{GDSAact}
defines an $\cA$-action covering the $\cA$-action on $\Bun_{\Z/2}(M, M^0)$.
\end{defn}
Identifying functions on $\Bun_{\Z/2}(M, M^0)$ with sections of the trivial line bundle, hence of
$\LDS'\to\Bun_{\Z/2}(M, M^0)$, a section $\psi$ is invariant under the $\cA$-action if and only if
$\psi\in\ker(\tH_v)$ for all $v\in\Delta^0(M)$; hence, by \cref{GDSrelns}, this identifies the ground states of the
GDS model for $M$ with the space $\Gamma(\LDS')^\cA$ of invariant sections of $\LDS'$. By \cref{quotientBun},
$\LDS'\to\Bun_{\Z/2}(M, M^0)$ descends to a (nonequivariant) line bundle $\LDS\to\Bun_{\Z/2}(M)$, and there is an
isomorphism $\Gamma(\LDS')^\cA\cong\Gamma(\LDS)$, so the space of ground states of the GDS model is isomorphic to
$\Gamma(\LDS)$.

\subsubsection{Computing the isomorphism type of $\LDS$}
Given a principal $\Z/2$-bundle $P\to M$, the action of $\Aut(P)$ on $(\LDS)_P$ is a character of $\Aut(P)$, and
the data of these characters for all $P\in\pi_0\Bun_{\Z/2}(M)$ determines $\LDS$ up to isomorphism. In this
section, we compute these characters, describing the answer in \cref{GDScharcomp}.

Let $P\to M$ be a principal $\Z/2$-bundle and $\phi\in\Aut(P)$. Let $\cV$ denote the set of vertices on which
$\phi$ is nontrivial, and order this set as $\set{v_1,\dotsc,v_m}$. Fix a trivialization $\xi_0$ of $P|_{M^0}$ and
let
\begin{equation}
\label{xiidefn}
	\xi_i\coloneqq \delta_{v_i}\cdot(\delta_{v_{i-1}}\cdot(\dotsb\cdot(\delta_{v_1}\cdot\xi_0)\dotsb)).
\end{equation}
In \cref{quotientBun}, we identified the action of $\phi$ on $\LDS$ with the action of $t_\phi$ on $\LDS'$, which
is multiplication by 
\begin{equation}
\label{sigcv}
\sigma_\cV\coloneqq \prod_{i=1}^m \sigma(v_i, (P,\xi_i)).
\end{equation}
To compare $\LDS$ and $L_\beta$, we need to pass from this description of $\sigma_\cV$ in terms of simplices to
a description only depending on $M$ and $P$. The following theorem makes this transition; afterwards we use
characteristic classes to finish the calculation.

As in \cref{quantchar}, let $P_\phi\to S^1\times M$ denote the mapping torus of $\phi$.
\begin{thm}
\label{signcalcPD}
Let $N\subset S^1\times M$ be an embedded submanifold representing the Poincaré dual to $\alpha(P_\phi)\in
H^1(S^1\times M;\Z/2)$. Then $\sigma_\cV = (-1)^{\chi(N)}$.
\end{thm}
Our proof has two parts.
\begin{enumerate}
	\item First, the simplicial part: we construct an $(n-1)$-cycle $C$ on $S^1\times M$, cellular with respect to
	a certain CW structure, which represents the Poincaré dual of $\alpha(P_\phi)$ (\cref{NPD}) and such that if
	$\abs C$ denotes the geometric realization of $C$, then $\sigma_\cV = (-1)^{\chi(\abs C)}$ (\cref{CEulercalc}).
	\item Then, we show that replacing $\abs C$ with a smoothly embedded representative of the homology class of
	$C$ does not change the mod 2 Euler characteristic (\cref{celltosmooth}).
\end{enumerate}
The proof employs the dual CW structure $\Pi^\vee$ to the given triangulation $\Pi$; see \cref{dualremark} for more
information. Let $S^1(m)$ denote the simplicial structure on $S^1$ with $m$ vertices, and choose an identification
of the vertices with $\Z/m$ such that $i$ and $i+1\bmod m$ share an edge for each $i$. Then let
$S^1(m)\times\Pi^\vee$ denote the product CW structure.

For any $i\in\Z/m$, the cellular $1$-cochain $\spin_{(P,\xi_i)}\colon\Delta^1(M;\Pi)\to\Z/2$ is a cocycle
representative for $\alpha(P)\in H^1(M;\Z/2)$, and therefore
\begin{equation}
	Y_i\coloneqq \set{e^\vee\mid e\in\Delta^1(M;\Pi)\text{ and } \spin_{(P,\xi_i)}(e) = 1}\subset
	\Delta^{n-2}(M;\Pi^\vee)
\end{equation}
is a cellular $(n-2)$-cycle representative for the Poincaré dual of $\alpha(P)$ in $H_{n-2}(M;\Z/2)$. From the
definitions of $Y_i$ and of $\xi_i$~\eqref{xiidefn} we see that
\begin{equation}
\label{Yii-1}
	Y_i = Y_{i-1} + \partial v_i^\vee,
\end{equation}
where $i-1$ is interpreted in $\Z/m$, and that
\begin{equation}
\label{bigCdefn}
	C\coloneqq \sum_{i\in\Z/m} \paren{(i, i+1)\times Y_i + \set i\times v_i^\vee} \subset\Delta^n(S^1\times M;
	S^1(m)\times\Pi^\vee)
\end{equation}
is a cellular $(n-1)$-cycle on $S^1\times M$.
\begin{defn}
\label{bktphi}
If $P\to M$ is a principal $\Z/2$-bundle over a closed manifold $M$, there is an isomorphism $\Aut(P)\to
H^0(M;\Z/2)$ sending $\phi\in\Aut(P)$ to the function on $\pi_0(M)$ which is $0$ on a connected component if $\phi$
is trivial there and $1$ if $\phi$ is nontrivial there. The image of $\phi\in\Aut(P)$ under this isomorphism is
denoted $[\phi]$.
\end{defn}
For example, if $x\in H^1(S^1;\Z/2)$ denotes the generator, then
\begin{equation}
\label{expand_phi}
	\alpha(P_\phi) = \alpha(P) + x[\phi] \in H^1(S^1\times M;\Z/2).
\end{equation}
\begin{lem}
\label{NPD}
The homology class $C$ represents is the Poincaré dual of $\alpha(P_\phi)\in H^1(S^1\times M;\Z/2)$.
\end{lem}
\begin{proof}
Recall that $Y_0\subset \Delta^{n-2}(M;\Pi^\vee)$ is a cellular $(n-2)$-cycle representing the Poincaré dual of
$\alpha(P)\in H^1(M;\Z/2)$. The $(n-1)$-cycle in $S^1\times M$ defined to be the set of $(n-1)$-cells of
\begin{equation}
	(S^1\times\abs{Y_0}) \cup \bigcup_{\substack{M_i\in\pi_0(M)\\{}[\phi](M_i) = 1}} \set 0\times M_i
\end{equation}
represents the Poincaré dual to $\alpha(P) + x[\phi] = \alpha(P_\phi)$~\eqref{expand_phi}, and is homologous to $C$
in $Z_{n-1}^{S^1(m)\times\Pi^\vee}(S^1\times M;\Z/2)$ by adding boundaries of the form $\partial((0,i)\times
v_i^\vee)$.
\end{proof}
\begin{lem}
\label{Y_iEulercalc}
For $1\le i\le m$, let $Z_{v_i}(P,\xi_i)$ be as in \cref{alternatesign}. Then $\#(\overline{Y_i}\cap\partial
v_i^\vee) = \#(Z_{v_i}(P,\xi_i))$ and therefore $(-1)^{1 + \chi(\abs{Y_i}\cap \partial v_i^\vee)} = \sigma(v_i,
(P,\xi_i))$.
\end{lem}
\begin{proof}
This is a matter of unwinding the definitions: $c\in\overline{Y_i}\cap\partial v_i^\vee$ means that $v_i\in\partial
c^\vee$ and either
\begin{enumerate}
	\item $c$ is an $(n-2)$-cell and $\spin_{(P, \xi_i)}(c^\vee) = 1$, or
	\item there is an $(n-2)$-cell $e\in Y_i$ with $c\in\partial e$, i.e.\ $\spin_{(P,\xi_i)}(e^\vee) = 1$ and
	$e^\vee\in\partial c^\vee$.
\end{enumerate}
These are exactly the conditions for $c^\vee$ to be in $Z_{v_i}(P, \xi_i)$, so $\#(\overline{Y_i}\cap\partial v_i)
= \#(Z_{v_i}(P,\xi_i))$, and the rest of the conclusion then follows from \cref{alternatesign}.
\end{proof}
\begin{prop}
\label{CEulercalc}
$(-1)^{\chi(\abs C)} = \sigma_\cV$.
\end{prop}
\begin{proof}
The projection map $\pi\colon S^1\times M\surj S^1$ is cellular with respect to $S^1(m)\times\Pi^\vee$ and
$S^1(m)$; if $D_i \coloneqq \abs C\cap\pi^{-1}([i, i+1))$, then each $D_i$ is a union of cells and
\begin{equation}
	\abs C = \coprod_{i\in\Z/m} D_i.
\end{equation}
Define $A_i$ and $B_i$ by $\pi^{-1}(\set i) = \set i\times A_i$ and $\pi^{-1}((i,i+1)) = (i,i+1)\times B_i$; $A_i$
and $B_i$ are also unions of cells. Then
\begin{subequations}
\begin{align}
	A_i &= \abs{Y_i}\cup\abs{Y_{i-1}}\cup\abs{v_i^\vee} = \abs{Y_i}\cup\abs{v_i^\vee}
\intertext{because $Y_{i-1} = Y_i + \partial v_i^\vee$~\eqref{Yii-1}, and}
	B_i &= \abs{Y_i}.
\end{align}
\end{subequations}
Therefore
\begin{equation}
\begin{aligned}
	\#(\text{cells of $D_i$}) &= \#(\text{cells of $A_i$}) + \#(\text{cells of $B_i$})\\
	&= \chi(\abs{Y_i}\cup\abs{v_i^\vee}) + \chi(\abs{Y_i})\\
	&= \chi(\abs{Y_i}\cup \operatorname{int}(\abs{v_i^\vee})\cup\abs{\partial v_i^\vee}) + \chi(\abs{Y_i})\\
	&= 1 + \chi(\abs{Y_i}) + \chi(\abs{\partial v_i^\vee}) - \chi(\abs{Y_i}\cap\abs{\partial v_i^\vee}) +
	\chi(\abs{Y_i})\\
	&\equiv_2 1 + \chi(\abs{Y_i}\cap\abs{\partial v_i^\vee}),
\end{aligned}
\end{equation}
since $\partial v_i^\vee\cong S^{n-1}$, which has even Euler characteristic. Looking at the definition of
$\sigma_\cV$ from \eqref{sigcv}, it suffices to equate $(-1)^{1+\chi(\abs{Y_i}\cap\partial v_i^\vee)}$ with
$\sigma(v_i, (P,\xi_i))$, which is taken care of by \cref{Y_iEulercalc}.
\end{proof}
Now we show that we can replace $\abs C$ with a smooth representative of the homology class of $C$.
\begin{defn}
Let $M$ be a smooth manifold and $r\in\Z_{\ge 0}\cup\set\infty$. A \term{$C^r$ triangulation} of $M$ is a
triangulation $(K, f\colon \abs K\to M)$ of $M$ such that for every simplex $e$ of $K$, $f|_{\abs e}$ is a $C^r$
map.
\end{defn}
\begin{thm}[{Munkres~\cite[Theorem 10.6]{Munkres}}]
\label{munkresthm}
Let $W$ be a compact manifold and $r\in\Z_{>0}\cup\set\infty$. Then every $C^r$ triangulation of $\partial W$
extends to a $C^r$ triangulation of $W$.
\end{thm}
\begin{cor}
\label{codim1iscell}
Let $X$ be a closed smooth manifold and $Y\subset X$ be a smooth codimension-one submanifold. Then there is a
triangulation of $X$ such that $Y$ is a union of simplices.
\end{cor}
\begin{proof}
Let $\nu\to Y$ denote the normal bundle of $Y\inj X$, $D(\nu)\to Y$ denote the unit disc bundle of $\nu$, and
$S(\nu)\to Y = \partial D(\nu)$ denote the unit sphere bundle of $\nu$. Using the tubular neighborhood theorem, we
choose an embedding $i\colon D(\nu)\inj M$ such that the original embedding of $Y$ in $X$ is the zero section of
$D(\nu)\to Y$ followed by $i$.

Let $r \ge 1$. Given a $C^r$ triangulation $\Pi(N)$ of $Y$, we can triangulate $D(\nu)$: let $\Pi(I)$ denote the
triangulation of $[-1,1]$ which has vertices precisely at the integers, which is a smooth triangulation. For any
simplex $e$ of $\Pi(Y)$, $D(\nu)|_{\abs e}$ is isomorphic to $\abs e\times [-1,1]$; choose an isomorphism $\psi_e$,
and give $D(\nu)_{\abs e}$ the product triangulation $\abs e\times\Pi(I)$. These are compatible as $e$ varies: if
$e'$ is another cell and $\abs{e'}$ intersects $\abs e$, $(\psi_{e'}^{-1}\circ\psi_e)|_{\abs e\cap\abs{e'}}$ is
either the identity or multiplication by $-1$ on the fiber. Both of these send simplices to simplices, so we can
glue the triangulations on $D(\nu)|_{\abs e}$ and $D(\nu)|_{\abs{e'}}$. Doing this for all simplices of $Y$ defines
a $C^r$ triangulation $\Pi(D(\nu))$ of $D(\nu)$ in which $Y\subset D(\nu)$ is a union of simplices.

This induces a $C^r$ triangulation of $S(\nu) = \partial(\overline{X\setminus D(\nu)})$, which by \cref{munkresthm}
extends to a triangulation of $\overline{X\setminus D(\nu)}$. We glue this triangulation to $\Pi(D(\nu))$, since
both triangulations agree on $S(\nu)$, to obtain a triangulation of $X$ in which $Y$ is a union of simplices.
\end{proof}
\begin{lem}
\label{changeface}
Let $\Pi$ be a triangulation of an $n$-manifold $X$, $C\in Z_{n-1}^\Pi(X;\Z/2)$, and $f\in\Delta^n(X)$. Then
\begin{equation}
	\chi(\abs C) \equiv \chi(\abs{C+\partial f})\bmod 2.
\end{equation}
\end{lem}
\begin{proof}
The sets of simplices in $\abs C$ and $\abs{C+\partial f}$ agree away from $\abs f$, so if $R_0\coloneqq \abs
C\cap\abs{\partial f}$ and $R_1\coloneqq \abs{C+\partial f}\cap\abs{\partial f}$, then it suffices to show
$\chi(R_0)\equiv\chi(R_1)\bmod 2$.

Inclusion-exclusion implies
\begin{equation}
\label{IEeuler}
	\chi(R_0) + \chi(R_1) \equiv \chi(\abs{\partial f}) + \chi(R_0\cap R_1)\bmod 2.
\end{equation}
Since $\abs{\partial f}\cong S^{d-1}$, its Euler characteristic is even. Next we show $R_0$ is a topological
manifold with boundary: if $R_0$ is empty or all of $\abs{\partial f}$, this is clear, and otherwise $R_0$ is an
iterated boundary connect sum of its $(n-1)$-simplices. Since $R_0\cap R_1 = \partial R_0$, $R_0\cap R_1$ is
null-bordant as a topological manifold, so its Euler characteristic is even, and~\eqref{IEeuler} simplifies to
$\chi(R_0) = \chi(R_1)\bmod 2$.
\end{proof}
\begin{prop}
\label{celltosmooth}
With $C$ as in~\eqref{bigCdefn}, if $N\inj S^1\times M$ is a smooth representative for the homology class of $C$
(namely, the Poincaré dual of $\alpha(P_\phi)$), then $\chi(\abs C) \equiv \chi(N)\bmod 2$.
\end{prop}
\begin{proof}
Let $\Pi_1$ be the barycentric subdivision of $\Pi$; as noted in \cref{dualremark}, this is also a ``refinement''
of $\Pi^\vee$, in that every cell of $\Pi^\vee$ is a union of simplices of $\Pi_1$. By \cref{codim1iscell}, there
is a triangulation $\Pi_t$ of $M$ such that $N$ is a union of simplices; let $\Pi'$ be a common refinement of
$\Pi_1$ and $\Pi_t$, and $S^1(m)\times\Pi'$ be the product triangulation of $S^1\times M$.

Let $\Ctop\in Z_{n-1}^{S^1(m)\times \Pi'}(S^1\times M;\Z/2)$ denote the cycle whose simplices are those contained
in the cells of $C$; then $\abs{\Ctop} = \abs C$. If $\Csm\in Z_{n-1}^{S^1(m)\times \Pi'}(S^1\times M;\Z/2)$
denotes the $(n-1)$-simplices in $N$, then $N = \abs{\Csm}$ and $\Ctop$ and $\Csm$ are homologous, so there are
$n$-cells $f_1,\dotsc,f_\ell$ such that
\begin{equation}
	\Csm = \Ctop + \sum_{i=1}^\ell \partial f_i.
\end{equation}
We apply \cref{changeface} $\ell$ times and conclude.
\end{proof}
By combining this with \cref{CEulercalc}, we have proven \cref{signcalcPD}.

Next, we translate $(-1)^{\chi(N)}$ into an expression involving characteristic classes of $M$ and $P$.
\begin{prop}
\label{charclass}
Let $M$ be a closed manifold, $P\to M$ be a principal $\Z/2$-bundle, and $N\subset M$ be a smoothly embedded,
codimension-$1$ submanifold representing the Poincaré dual to $\alpha(P)$. Then,
\begin{equation}
	\chi(N)\bmod 2 = \ang*{\frac{w(X)\alpha(P)}{1+\alpha(P)}, [X]}.
\end{equation}
\end{prop}
But before we prove this:
\begin{lem}
\label{no_longer_gysin}
Let $L\to X$ be a line bundle over a closed manifold $X$ and $Y\inj X$ be a smoothly embedded closed submanifold
representing the Poincaré dual to $w_1(L)$, with normal bundle $\nu\to Y$. Then, as line bundles over $Y$,
$\nu\cong L|_Y$.
\end{lem}
\begin{proof}
If $i_!\colon H^*(Y;\Z/2)\inj H^{*+1}(X;\Z/2)$ denotes the Gysin map (which is Poincaré dual to restriction
$H^*(X;\Z/2)\to H^*(Y;\Z/2)$), then $i_!(1)$ is Poincaré dual to $[Y]\in H_{d-1}(X;\Z/2)$ and $i^*i_!(1) =
w_1(\nu)$. By construction, $[Y]$ is Poincaré dual to $w_1(L)$, so $i^*w_1(L) = w_1(L|_Y) = w_1(\nu)$.  As line
bundles are classified by their Stiefel-Whitney classes, $\nu \cong L|_Y$.
\end{proof}
\begin{proof}[Proof of \cref{charclass}]
Let $j\colon N\inj M$ be inclusion. Since $N$ represents the Poincaré dual of $\alpha(P)$, then for any $x\in
H^{n-1}(M;\Z/2)$,
\begin{equation}
	\ang{j^*x, [N]} = \ang{\alpha(P)x, [M]}.
\end{equation}
We will use this to carry the mod 2 Euler characteristic of $N$, which is equal to $\ang{w(N), [N]}$, to the
cohomology of $M$; in order to do so, we must show $w(N)\in\Im(j^*)$.

If $\nu\to N$ denotes the normal bundle of $N$, there is a short exact sequence of vector bundles on $N$
\begin{equation}\xymatrix{
	0\ar[r] & TN\ar[r] & j^*TM\ar[r] & \nu\ar[r] & 0,
}\end{equation}
so $w(j^*TM) = j^*w(M) = w(N)w(\nu)$. Since $\nu$ is a line bundle,
\begin{equation}
	w(\nu) = 1 + w_1(\nu) = 1+j^*\alpha(P) = j^*(1+\alpha(P))
\end{equation}
by \cref{no_longer_gysin}. Hence
\begin{equation}
	j^*w(M) = w(N) j^*(1+\alpha(P)).
\end{equation}
Since $\alpha(P)\in H^*(X;\Z/2)$ is nilpotent, $j^*(1+\alpha(P))$ is invertible, and therefore
\begin{equation}
	w(N) = \frac{j^*w(M)}{j^*(1+\alpha(P))} = j^*\paren{\frac{w(M)}{1+\alpha(P)}}.
\end{equation}
Thus we can invoke Poincaré duality:
\begin{equation}
	\chi(N)\bmod 2 = \ang{w(N), [N]} = \ang*{\alpha(P)\cdot\frac{w(M)}{1+\alpha(P)}, [M]}. \qedhere
\end{equation}
\end{proof}
Combining this with \cref{signcalcPD}, we get:
\begin{cor}
\label{GDScharcomp}
If $P\in\Bun_{\Z/2}(M)$, the character of $\Aut(P)$ acting on $(\LDS)_P$ has $\phi$ act by multiplication by
\begin{equation}
\label{LDScharanswer}
	(-1)^{\ang{\alpha(P_\phi)w(S^1\times M)/(1+\alpha(P_\phi)), [S^1\times M]}}\in\set{\pm 1}\subset\C^\times.
\end{equation}
\end{cor}
Next, we compare this with the character of $\Aut(P)$ acting on $(L_\beta)_P$ and conclude.
\begin{proof}[Proof of \cref{mainthm}]
\Cref{quantchar} tells us that in the character of $\Aut(P)$ acting on $(L_\beta)_P$, $\phi$ acts by
$\Zcl_\beta(S^1\times M, P_\phi)$; by \cref{classpartdefnthm}, this is exactly~\eqref{LDScharanswer}. Hence
$\LDS\cong L_\beta$.
\end{proof}
\begin{rem}
\label{another_LE}
Suppose $n$ is even, and let $Z_2\colon\Bord_n\to\Vect_\C$ denote the quantum gauge-gravity TQFT with Lagrangian
$\beta_2\coloneqq w\alpha/(1+\alpha^2)$. Then $\DS(\RP^n) = 1$ and $Z_2(\RP^n) = 0$, so $\DS\ne Z_2$. However, a
characteristic-class computation shows that for any closed $(n-1)$-manifold $M$, there is an isomorphism
$\DS(M)\cong Z_2(M)$ equivariant with respect to the natural $\MCG(M)$-actions on the state spaces.\footnote{We
will recall the definition of these $\MCG(M)$-actions in \S\ref{TQFTMCG}.} This means that in the sense of
\cref{LEEFT}, both $\DS$ and $Z_2$ capture the ground states of the GDS model, and that it is not clear how to
distinguish them using data from the lattice. In physics, however, the low-energy effective theory of a lattice
model is expected to be unique.

Freed-Hopkins~\cite[\S7.3]{FreedHopkins}, following Kong-Wen~\cite{KongWen}, suggest that the low-energy effective
theory may only be defined on manifolds which locally have a direction of time, i.e.\ manifolds $M$ together with a
reduction of the structure group of $TM$ from $\O_n$ to $\O_{n-1}$. That is, it should be possible to calculate the
partition function on such manifolds using locality of the lattice model, and it might not be possible to calculate
further in general. Alternatively, Shapourian-Shiozaki-Ryu~\cite{SSR17} describe a method to compute partition
functions on $\RP^2$ for 2D SPTs defined by a Hamiltonian, and it is possible their method would generalize, though
we have not pursued this.
\end{rem}

	\subsection{Mapping class group representations}
		\label{mcg}
If $Z\colon\Bord_n\to\Vect_\C$ is an $n$-dimensional TQFT and $M$ is a closed $(n-1)$-manifold, the mapping class
group of $M$ naturally acts on $Z(M)$, as we describe below. We will define a similar $\Diff(M)$-action on the
ground states of the GDS model on $M$ and show that the isomorphism $\DS(M)\cong L(M)$ is an isomorphism of
$\Diff(M)$-representations.  Since the former representation is trivial when restricted to the connected component
of the identity in $\Diff(M)$, the latter is too; thus the action on the ground states of the GDS model is also an
$\MCG(M)$-action.

\subsubsection{The mapping class group action for TQFTs}
\label{TQFTMCG}
Let $\Diff(M)$ denote the diffeomorphism group of $M$ and $\Diff_0(M)$ denote the connected component of the
identity, so that $\MCG(M) = \Diff(M)/\Diff_0(M)$. For any $\vp\in\Diff(M)$, let $C_\vp$ denote the \term{mapping
cylinder} of $\vp$, i.e.\ the cobordism $[0,1]\times M$ from $M$ to itself, where $M$ is attached via the identity
at $0$ and via $\vp$ at $1$.

If $Z\colon\Bord_n\to\Vect_\C$ is a TQFT, then the assignment $\vp\mapsto Z(C_{\vp})\colon Z(M)\to Z(M)$ defines
an action of $\Diff(M)$ on $Z(M)$. If $\vp\in\Diff_0(M)$, then there is a smooth isotopy $\vp_t\colon [0,1]\times
M\to M$ such that $\vp_t(0,x) = x$ and $\vp_t(1, x) = \vp(x)$, and in particular there is a diffeomorphism of
cobordisms $C_\id\cong C_\vp$ defined by the map
\begin{equation}
\begin{aligned}
	[0,1]\times M &\to [0,1]\times M\\
	(t, x)		  &\mapsto (t, \vp_t(x)).
\end{aligned}
\end{equation}
Therefore $Z(C_\vp) = Z(C_\id) = \id$, so this $\Diff(M)$-action factors through $\Diff_0(M)$ to an
$\MCG(M)$-action on $Z(M)$.

\subsubsection{The $\Diff(M)$-action for a lattice model}
\label{genMCG}
We will imitate the first half of the above argument for a lattice model with some assumptions, constructing a
$\Diff(M)$-action on the space of ground states of the model on $M$; in \S\S\ref{MCG_TC} and \ref{MCG_GDS}, we will
see these factor through $\Diff_0(M)$ and define actions of the mapping class group on the spaces of ground states
of the toric code and GDS models.

We require the following of our lattice model.
\begin{enumerate}[label={(A\arabic*)}, ref={A\arabic*}]
	\item\label{latcond} The model is defined for closed $(n-1)$-manifolds equipped with a lattice, which here
	means a CW structure or a triangulation, or one of these structures subject to some condition that can be
	satisfied on all closed $(n-1)$-manifolds and for which any two such structures on a manifold admit a common
	refinement.
	\item\label{dataisom} Data of, for every refinement $\Pi\to\Pi'$ of lattices, an isomorphism from the space of
	low-energy states of the model on $\Pi$ to the space of low-energy states of the model on $\Pi'$, which is
	functorial under composition of refinements.
\end{enumerate}
Examples of conditions satisfying the constraint in~\eqref{latcond} include regular CW complexes and the class of
smooth triangulations we considered when defining the GDS model.

With these assumptions in place, let $\Lat(M)$ denote the poset category of lattices on a closed manifold $M$,
where morphisms are refinements. Then~\eqref{dataisom} defines a functor $L\colon\Lat(M)\to\Vect_\C$; let
$Z(M)\coloneqq\varinjlim L$. For every lattice $\Pi$ on $M$ there is a canonical isomorphism
$c_\Pi\colon L(\Pi)\overset\cong\to Z(M)$. We will construct a $\Diff(M)$-action on $Z(M)$.

Given a lattice $\Pi$ and $f\in\Diff(M)$, we get a new lattice $f(\Pi)$ by postcomposing the attaching maps in
$\Pi$ with $f$, and this defines a $\Diff(M)$-action on $\Lat(M)$, i.e.\ a functor $\pt/\Diff(M)\to \Lat(M)$. We
think of $L$ as a vector bundle over the category $\Lat(M)$ and $Z(M)$ as its space of sections. The fibers of this
vector bundle over two lattices $\Pi$ and $\Pi'$ are canonically identified by $c_{\Pi'}^{-1}\circ c_\Pi$;
therefore we can lift the $\Diff(M)$-action on the base to make $L$ into a $\Diff(M)$-equivariant vector bundle in
the trivial way: given $x\in L(\Pi)$ and $f\in\Diff(M)$, $f(x) \coloneqq c_{f(\Pi)}^{-1}\circ c_{\Pi}(x)$. The
space of sections of an equivariant vector bundle has an induced action: explicitly, given $x\in Z(M)$ and an
$f\in\Diff(M)$, choose a lattice $\Pi$; then $f(x) = c_{f(\Pi)}(c_{\Pi}^{-1}(x))$, and this does not depend on the
choice of $\Pi$.


\subsubsection{The $\MCG(M)$-action for the toric code}
\label{MCG_TC}
Consider the $n$-dimensional toric code, which is formulated on closed $(n-1)$-manifolds with a CW structure. A
refinement $\vp\colon\Xi\to\Xi'$ of CW structures on $M$ induces a pullback map
\begin{equation}
	\vp^*\colon\Bun_{\Z/2}(M^1_{\Xi'}, M^0_{\Xi'})\to \Bun_{\Z/2}(M^1_\Xi, M^0_\Xi).
\end{equation}
hence a pushforward map on state spaces: $\vp_*\colon\cH(\Xi)\to\cH(\Xi')$.
\begin{rem}
\label{naiveTC}
The pushforward $\vp_*$ does not restrict to an isomorphism on the spaces of ground states. Consider the refinement
$\Xi\to\Xi'$ in \cref{noground} and $(P,\xi)$ which induce the indicated spins on the $1$-cells of $\Xi'$. If $f$
is a ground state for $\Xi'$, it must vanish on $(P,\xi)$, because $(P,\xi)$ has nontrivial holonomy around the
boundaries of the pictured $2$-cells, but pulled back to $\Xi$, this is no longer the case. Therefore $\Im(\vp_*)$
contains states which do not vanish on $(P,\xi)$, hence are not ground states.
\end{rem}
The issue is that functions in the image of $\vp_*$ may not vanish on bundles with nontrivial holonomy around
certain boundaries of $2$-cells, so in order to satisfy~\eqref{dataisom}, we zero out their values on any such
bundle. Let $\cP\colon\cH_{\Xi'}\to\cH_{\Xi'}$ denote this projection: that is, if $f\in\cH_{\Xi'}$ and $(P,
\xi)\in\Bun_{\Z/2}(M^1_{\Xi'}, M^0_{\Xi'})$, let
\begin{equation}
\label{rightprojector}
	(\cP f)(P, \xi) \coloneqq \begin{cases}
		f(P, \xi), &\text{if $\Hol_P(e) = 0$ for all $e\in\Delta^2(M;\Xi')$},\\
		0, &\text{otherwise.}
	\end{cases}
\end{equation}

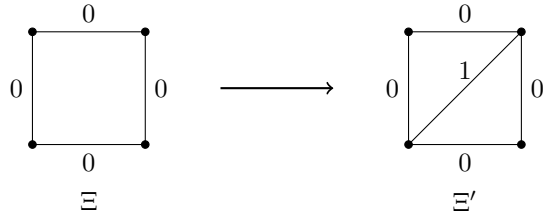
\begin{figure}[h!]
\begin{tikzpicture}
\coordinate (r0) at (0,0);
\coordinate (r3) at (0,1.5);
\coordinate (r1) at (1.5,0);
\coordinate (r2) at (1.5,1.5);

\coordinate (s0) at (5,0);
\coordinate (s3) at (6.5,0);
\coordinate (s1) at (5,1.5);
\coordinate (s2) at (6.5,1.5);

\draw[->, thick] (2.5, 0.75) -- (4, 0.75);

\draw (r0) -- node[below]{$0$} (r1) -- node[right]{$0$} (r2) -- node[above]{$0$} (r3) -- node[left]{$0$} (r0);
\draw (s0) -- node[left]{$0$} (s1) -- node[above]{$0$} (s2) -- node[right]{$0$} (s3) -- node[below]{$0$} (s0);
\draw (s0) -- node[above]{$1$} (s2);

\tikzpt{(r0)};
\tikzpt{(r1)};
\tikzpt{(r2)};
\tikzpt{(r3)};

\tikzpt{(s0)};
\tikzpt{(s1)};
\tikzpt{(s2)};
\tikzpt{(s3)};

\node[below] at (0.75, -0.5) {$\Xi$};
\node[below] at (5.75, -0.5) {$\Xi'$};
\end{tikzpicture}
\caption{Consider a refinement $\Xi\to\Xi'$ of CW structures as above, together with a
$(P,\xi)\in\Bun_{\Z/2}(M^1_{\Xi'}, M^0_{\Xi'})$ such that the labels on the $1$-simplices represent
$\spin_{(P,\xi)}$, as in \cref{noCW}. In \cref{naiveTC}, we discuss how $(P,\xi)$ illustrates a subtlety in
defining the map from the ground states of the toric code for $\Xi$ to those on $\Xi'$.}
\label{noground}
\end{figure}

\begin{lem}
\label{TCgroundground}
The map $\cP\circ \vp_*$ sends ground states to ground states, hence restricts to an isomorphism
$L(\Xi)\overset\cong\to L(\Xi')$ functorial in the sense of~\eqref{dataisom}.
\end{lem}
\begin{proof}
Let $f\in L(\Xi)$. By construction $\cP(\vp_*(f))$ vanishes on principal $\Z/2$-bundles with nontrivial holonomy,
so it suffices to check that it does not depend on the trivializations on the $0$-cells. This is not changed by
$\cP$, so we can just think about $\vp_*(f)$. Let $v\in\Delta^0(M, \Xi')$ and suppose $v$ is also a $0$-cell of
$\Xi$. Then $\vp_*(f)$ cannot depend on the trivialization at $v$, because $f$ does not depend on the
trivialization at $v$. If instead $v$ is not a $0$-cell of $\Xi$, so is created by the refinement, then $\vp_*(f)$
also does not depend on the trivialization at $v$, because $\vp_*(f)(P,\xi)$ is computed by pulling back to $\Xi$,
where $v$ is not a cell.
\end{proof}

Therefore the argument of \S\ref{genMCG} applies to define for any closed $(n-1)$-manifold $M$ an action of
$\Diff(M)$ on the ground states of the toric code. Under the identification $Z(M)\cong \C[\Bun_{\Z/2}(M)]$, this
representation is the one induced from the usual $\Diff(M)$-action on $\pi_0\Bun_{\Z/2}(M)\cong H^1(M;\Z/2)$, which
factors through $\Diff_0(M)$ to define an action of the mapping class group.

Recall from \S\ref{LETC} that the $\Z/2$-Dijkgraaf-Witten theory with Lagrangian equal to $0$, denoted $\DW$,
captures the ground states of the toric code. This theory assigns to a closed $(n-1)$-manifold $M$ the vector space
$\DW(M)\cong\C[\Bun_{\Z/2}(M)]$, and to a cobordism assigns a push-pull map, which implies that the
$\MCG(M)$-action on $\DW(M)$ is also the action induced from the standard action on $\pi_0\Bun_{\Z/2}(M)$.
Therefore we have proved the following extension of \cref{TC_deriv}.
\begin{thm}
The identification of the space of ground states of the toric code for $M$ with $\DW(M)$ in \cref{TC_deriv} is
equivariant with respect to the $\MCG(M)$-actions on both sides.
\end{thm}
The mapping class group action determines the partition functions of mapping tori: if $f\in\MCG(M)$, then $Z(M_f)$
is the trace of $f$ acting on $Z(M)$. Though we can see these partition functions from the lattice, it is not clear
in general how to extend this to arbitrary closed $n$-manifolds.

\subsubsection{The $\MCG(M)$-action for the GDS model}
\label{MCG_GDS}
Let $\Cell(M)$ denote the poset category whose objects are smooth triangulations on $M$ such that the 0-clopen star
of every vertex is contractible, and whose morphisms are refinements, and let $\vp\colon \Pi\to\Pi'$ be such a
refinement. Define $\vp_*$ and $\cP$ as in the previous section, and let $\cP'\colon\cH_{\Pi'}\to\cH_{\Pi'}$ be the
projection onto $\bigcap_v \tH_v$ which is orthogonal with respect to the inner product in which the
$\delta$-functions on elements of $\pi_0\Bun_{\Z/2}(M^1, M^0)$ are an orthonormal basis.
\begin{lem}
The map $\cP\circ\cP'\circ \vp_*$ sends ground states to ground states, hence restricts to an isomorphism
$L(\Pi)\overset\cong\to L(\Pi')$ functorial as in~\eqref{dataisom}.
\end{lem}
\begin{proof}
Suppose $\vp$ adds no $0$-simplices and $1$-simplices to $\Pi$, so $\cH_{\Pi'}\cong\cH_{\Pi'}$ and $\vp_*$ is the
identity. Then $\vp$ adds no cells at all, because it is not possible to add cells to a manifold that is a
simplicial complex without adding $0$- or $1$-simplices, so $\vp$ is the identity refinement and the lemma follows
because $\cP$ and $\cP'$ are projections.

If otherwise, we show that $\vp_*$ of a nonzero ground state is not a ground state, so that the orthogonal
projection thereafter sends it to a nonzero ground state. If $\vp$ adds any $1$-simplices to $\Pi$ that do not
arise from splitting preexisting $1$-simplices into smaller ones, the construction in \cref{naiveTC} shows that
$\vp_*$ of a nonzero ground state is not a ground state; if the only $1$-simplices it adds are split from
preexisting ones, then it must add a $0$-simplex. If $\vp$ adds any $0$-simplices to $\Pi$, it must add a
$1$-simplex that is not split from a preexisting $1$-simplex, because all $0$-simplices must be trivalent.
\end{proof}
Therefore the argument of \S\ref{genMCG} applies to define for any closed $(n-1)$-manifold $M$ an action of
$\Diff(M)$ on the ground states of the GDS model. Under the identification of $Z(M)$ with the space of functions on
the set of $P\in \pi_0\Bun_{\Z/2}(M)$ such that $\ang{\alpha(P)w(M)/(1+\alpha(P)), [M]} = 0$, this representation
is the one induced from the usual $\Diff(M)$-action on this space, which is an invariant subspace of
$\C[\Bun_{\Z/2}(M)]$, just as in the previous section, and once again this factors through $\Diff_0(M)$ to define
an $\MCG(M)$-action.

Recall from \S\ref{LETC} that $\DS$ captures the ground states of the GDS model; using the push-pull map $\DS$
assigns to a cobordism, its $\MCG(M)$-action is the same, again induced from the standard action on
$\pi_0\Bun_{\Z/2}(M)$. Therefore \cref{mainthm} strengthens to the following statement.
\begin{thm}
The identification of the space of ground states of the GDS model for $M$ with $\DS(M)$ in \cref{mainthm} is
equivariant with respect to the $\MCG(M)$-actions on both sides.
\end{thm}
Again, this means we can see the partition functions of mapping tori from the lattice, but not of other
closed $n$-manifolds.

\section{Calculations}
	\label{calcs}
	In this section, we perform some calculations with the GDS Lagrangian in order to understand when $\DS$ is
isomorphic to a $\Z/2$-Dijkgraaf-Witten theory. First, we fix some notation.
\begin{itemize}
	\item Recall that $\alpha$ denotes the generator of $H^1(B\Z/2;\Z/2)\cong\Z/2$; in particular, it defines a
	characteristic class for principal $\Z/2$-bundles by pullback, and if $P\in\Bun_{\Z/2}(X)$, this characteristic
	class evaluated on $P$ is denoted $\alpha(P)\in H^1(X;\Z/2)$.
	\item $\DW\colon\Bord_n\to\Vect_\C$ denotes $\Z/2$-Dijkgraaf-Witten theory with the zero Lagrangian and
	$Z_{\alpha^n}\colon\Bord_n\to\Vect_\C$ denotes $\Z/2$-Dijkgraaf-Witten theory with Lagrangian $\alpha^n\in
	H^n(B\Z/2;\Z/2)$.
	\item Recall from \cref{bktphi} that if $P\to M$ is a principal $\Z/2$-bundle, the image of $\phi\in\Aut(P)$
	under the isomorphism $\Aut(P)\to H^0(M;\Z/2)$ is denoted $[\phi]$. Letting $x\in H^1(S^1;\Z/2)$ denote the
	generator, $\alpha(P_\phi) = \alpha(P) + x[\phi]$ in $H^*(S^1\times M;\Z/2)$.
\end{itemize}

We begin with a few example calculations. We will call a principal $\Z/2$-bundle $P\to M$ \term{permitted} if the
GDS action $\ang{w(M)\alpha(P_\phi)/(1+\alpha(P_\phi)), [M]}$ vanishes for all $\phi\in\Aut(P)$; thus $\DS(M)$
is the space of functions on the set of isomorphism classes of permitted bundles.
\begin{prop}
\label{eulerchartest}
If $M$ is a closed $(n-1)$-manifold, then the trivial bundle $\Ptriv\to M$ is permitted if and only if $\chi(M)$ is
even.
\end{prop}
\begin{proof}
The action for $\Ptriv$ and $\phi\in\Aut(\Ptriv)$ is
\begin{align}
	\ang*{\frac{w(M)\alpha((\Ptriv)_\phi)}{1 + \alpha((\Ptriv)_\phi)}, [S^1\times M]} &= \ang*{\frac{w(M)(x[\phi] +
	\alpha(\Ptriv))}{1 + (x[\phi] + \alpha(\Ptriv))}, [S^1\times M]}
	\intertext{by~\eqref{expand_phi}. Since $\Ptriv$ is trivial, $\alpha(\Ptriv) = 0$, so}
	&= \ang*{\frac{w(M)x[\phi]}{1 + x[\phi]}, [S^1\times M]}.
	\intertext{Since $(x[\phi])^2\in H^2(S^1;\Z/2) = 0$,}
	&= \ang{w(M)x[\phi], [S^1\times M]},
	\intertext{which by a Fubini theorem is}
	&= \ang{x[\phi], [S^1]}\ang{w(M), [M]}.
	\label{clact}
\end{align}
If $\phi$ is nontrivial, $\ang{x[\phi], [S^1]} = 1$. Hence the action is zero for all $\phi$ if and only if
$\ang{w(M), [M]}$, which is $\chi(M)$ mod 2, vanishes.
\end{proof}
\begin{cor}
\label{simpconn}
Let $M$ be simply connected. Then,
\begin{equation}
\begin{gathered}
\DS(M) \cong \begin{cases}
	0, &\chi(M)\text{\rm{} odd}\\
	\C, &\chi(M)\text{\rm{} even.}
\end{cases}
\end{gathered}
\end{equation}
\end{cor}
\begin{proof}
All principal $\Z/2$-bundles over such a manifold are trivial, so we just have to check whether the trivial bundle
is permitted.
\end{proof}
It is worth comparing this to the $\alpha^n$ Dijkgraaf-Witten theory.
\begin{lem}
\label{TDWtrivsect}
If $n > 1$ and $M$ is a closed $(n-1)$-manifold, $\Zcl_{\alpha^n}(S^1\times M, (\Ptriv)_\phi) = 0$ for any
automorphism $\phi$. In particular, if $M$ is simply connected, $Z_{\alpha^n}(M)\cong\C$.
\end{lem}
\begin{proof}
Let $\phi\in\Aut(\Ptriv)$, so
\begin{equation}
	\alpha((\Ptriv)_\phi) = \alpha(\Ptriv) + x[\phi] = x[\phi].
\end{equation}
The action is
\begin{equation}
	\ang{\alpha(P_\phi)^n, [S^1\times M]} = \ang{(x[\phi])^n, [S^1\times M]} = 0.\qedhere
\end{equation}
\end{proof}
\begin{prop}
\label{cprp}
\begin{equation}
\begin{gathered}
\DS(\CP^n\times\RP^2) \cong \begin{cases}
	\C, & n\text{\rm{} even}\\
	\C^2, &n\text{\rm{} odd.}
\end{cases}
\end{gathered}
\end{equation}
\end{prop}
\begin{proof}
Let $X \coloneqq \CP^n\times\RP^2$, and let $z$ be the generator of $H^1(X;\Z/2)\cong\Z/2$. Since
\begin{equation}
\chi(X) = \chi(\CP^n)\chi(\RP^2) = \begin{cases}
	0\bmod 2, &n\text{ odd}\\
	1\bmod 2, &n\text{ even,}
\end{cases}
\end{equation}
then by \cref{eulerchartest}, the trivial bundle is permitted if and only if $n$ is odd.

The other isomorphism class of principal $\Z/2$-bundles on $X$ is the one whose total space is the universal cover
of $X$, which we denote $P$. Then $\alpha(P) = z$, and for $\phi\in\Aut(P)$, the Lagrangian for $S^1\times X$ and $P_\phi$ is
\begin{align}
	\frac{\alpha(P_\phi) w(S^1\times X)}{1 + \alpha(P_\phi)} &= \frac{(z + x[\phi])w(\RP^2)w(\CP^n)}{1+ z +
	x[\phi]}.
	\intertext{Since $z + x[\phi]$ is nilpotent, $1+z+x[\phi]$ is invertible, so}
	&= \frac{(z+x[\phi])w(\RP^2)w(\CP^n)(1+z+x[\phi])}{(1+z+x[\phi])^2}.
	\intertext{Since $(x[\phi])^2 = 0$,}
	&= \frac{(1+z)^3(z+z^2+x[\phi]) w(\CP^n)}{1+z^2}\\
	&= (1+z)(z+z^2+x[\phi])w(\CP^n)\\
\label{nofubini}
	&= (z + x[\phi] + zx[\phi])w(\CP^n).
\end{align}
We want to pair this with $[S^1\times X]$, but~\eqref{nofubini} has no terms of degree $\dim(S^1\times X) = 2n+3$.
Thus
\begin{equation}
	\ang{(z+x[\phi] + zx[\phi])w(\CP^n), [S^1\times X]} = 0,
\end{equation}
so this bundle is always permitted.
\end{proof}
\begin{prop}
\label{RPnstates}
For $n\ge 2$,
\begin{equation}
\DS(\RP^n) \cong\begin{cases}
	\C, &n\text{\rm{} even}\\
	\C^2, &n\text{\rm{} odd.}
\end{cases}
\end{equation}
\end{prop}
\begin{proof}
Let $z\in H^1(\RP^n;\Z/2)$ denote the generator. By \cref{eulerchartest}, the trivial principal $\Z/2$-bundle is
permitted if and only if $n$ is odd. The other isomorphism class of principal $\Z/2$-bundles is the universal cover
$S^n\surj\RP^n$, with $\alpha(S^n) = z$, so it suffices to prove this bundle is always permitted. Let $\phi$ be an
automorphism of this principal bundle. The action is
\begin{align}
	\frac{\alpha(S^n_\phi)w(\RP^n)}{1+\alpha(S^n_\phi)} &= \frac{(z+x[\phi])(1+z)^{n+1}}{1+z+x[\phi]}.
	\intertext{Again, $z+x[\phi]$ is nilpotent, so $1+z+x[\phi]$ is invertible, so}
	&= \frac{(z+x[\phi])(1+z)^{n+1}(1+z+x[\phi])}{(1+z+x[\phi])^2}\\
	&= \frac{(1+z)^{n+1}(z+z^2+x[\phi])}{(1+z)^2}\\
\label{xphi}
	&= (1+z)^{n-1}(z+z^2+x[\phi]).
\end{align}
But in~\eqref{xphi}, only the $(1+z)^{n-1}z^2$ term contributes anything of degree $\dim(S^1\times\RP^n) = n+1$,
and this lives in $H^{n+1}(\RP^n;\Z/2)\otimes H^0(S^1;\Z/2)$, hence must be $0$. Thus~\eqref{xphi} has no terms of
top degree, so
\begin{equation}
\ang{(1+z)^{n+1}(z+z^2+x[\phi]), [S^1\times\RP^n]} = 0,
\end{equation}
and this bundle is always permitted.
\end{proof}

We now compare $\DS$ with $\Z/2$-Dijkgraaf-Witten theories.
\begin{lem}
\label{CCL1}
Let $M$ be a closed $(2k+1)$-manifold and $y\in H^1(M;\Z/2)$. Then $w_1(M)y^{2k} = 0$.
\end{lem}
\begin{proof}
Let $v_1$ denote the first Wu class. Then,
\begin{equation}
	w_1y^{2k} = v_1y^{2k} = \Sq^1((y^k)^2) = 0. \qedhere
\end{equation}
\end{proof}
\begin{thm}
\label{dim3isom}
In dimension $3$, $\DS$ is isomorphic to $Z_{\alpha^3}$.
\end{thm}
\begin{proof}
This follows from \cref{w_1isom} after observing
\begin{equation}
\label{3dTDW}
	(\alpha+w_1)^3 = \alpha^3 + w_1\alpha^2 + w_1^2\alpha + w_1^3.
\end{equation}
On any closed 3-manifold, $w_1^3 = 0$ because all closed 3-manifolds bound, and $w_1\alpha^2 =0$ by \cref{CCL1}.
Thus~\eqref{3dTDW} agrees with the Lagrangian for $\DS$.
\end{proof}
The relationship in dimension 3 between the double semion model and the $\Z/2$-Dijkgraaf-Witten theory with
Lagrangian $\alpha^3$ is known to physicists (see, e.g.,~\cite[\S II]{WW15}), though not previoiusly proven in this
form.


\begin{thm}
\label{isDW}
For even $n$, $\DS$ is isomorphic to $\DW$.
\end{thm}
\begin{proof}
By \cref{Wuisom}, it suffices to prove that $w(M)\alpha/(1+\alpha) = 0$ for any even-dimensional manifold $M$ and
$\alpha\in H^1(M;\Z/2)$. In \cref{charclass}, we saw $\ang{w(M)\alpha/(1+\alpha), [M]}$ is the mod 2 Euler
characteristic of a submanifold $N$ representing the Poincaré dual of $\alpha$. Since $N$ is a closed,
odd-dimensional manifold, its mod 2 Euler characteristic vanishes, so $w(M)\alpha/(1+\alpha) = 0$.
\end{proof}
\cite[Thm.~5.3]{FreedmanHastings} proved this for state spaces, and the proof idea is the same.

\begin{thm} 
\label{notDW}
For odd $n\ge 4$, $\DS$ is not isomorphic to any $\Z/2$-Dijkgraaf-Witten theory.
\end{thm}
\begin{proof}
By \cref{all_quantum_TDW}, it suffices to prove that $\DS$ is not isomorphic to $\DW$ and $Z_{\alpha^n}$.

If $n= 4k+1$ for some $k\ge 1$, then $\DS(\CP^{2k}) = 0$ by Corollary~\ref{simpconn}, but $\DW(\CP^{2k}) \cong \C$,
and $Z_{\alpha^n}(\CP^{2k})\cong\C$ by \cref{TDWtrivsect}.

If $n = 4k+3$ for some $k\ge 1$, then $\DS(\CP^{2k}\times\RP^2) \cong \C$ by Proposition~\ref{cprp} and
$\DW(\CP^{2k}\times\RP^2) \cong\C^2$. For the theory with Lagrangian $\alpha^n$, \cref{TDWtrivsect} gives us one
copy of $\C$ from the trivial bundle. If $P\to\CP^{2k}\times\RP^2$ denotes the nontrivial bundle and $z\in
H^1(\RP^2;\Z/2)$ denotes the generator, then $\alpha(P) = z$. For any $\phi\in\Aut(P)$,
\begin{align}
	\ang{\alpha(P_\phi)^n, [S^1\times\CP^{2k}\times\RP^2]} &= \ang{(z + x[\phi])^n,
	[S^1\times\CP^{2k}\times\RP^2]}.
	\intertext{Since $(x[\phi])^2 = 0$, this is}
	&= \ang{z^n + nz^{n-1}x[\phi], [S^1\times\CP^{2k}\times\RP^2]},
\end{align}
and since $z^3 = 0$, this is $0$. Thus the state space picks up another factor of $\C$, and
$Z_{\alpha^n}(\CP^{2k}\times\RP^2)\cong\C^2$.
\end{proof}
This was also proven in~\cite[Thm.~8.1]{FreedmanHastings}, with the same manifolds as counterexamples.

\newcommand{\etalchar}[1]{$^{#1}$}

\end{document}